%% file: main.tex
\documentclass[a4paper,10pt,romanappendices
]{IEEEtran}

\usepackage{commath,amsmath,graphicx,epstopdf,amsthm,amssymb,float,cite,multicol,enumitem,xcolor}

\usepackage{acronym}
\usepackage{amssymb,amsfonts}
\usepackage{algorithmic}
\usepackage{textcomp}
\usepackage{gensymb}
\usepackage{subcaption}

\acrodef{l.h.s.}[l.h.s.]{left hand side}
\acrodef{r.h.s.}[r.h.s.]{right hand side}
\acrodef{w.p.}[w.p.]{with probability}
\acrodef{MSE}[MSE]{mean-squared error}
\acrodef{MAP}[MAP]{maximum \emph{a-posteriori} probability}
\acrodef{RMSE}[RMSE]{root-MSE}
\acrodef{ULA}[ULA]{uniform linear array}
\acrodef{MIMO}[MIMO]{multiple-input multiple-output}
\acrodef{MMSE}[MMSE]{{minimum mean-squared-error}}
\acrodef{ADC}[ADC]{analog-to-digital converters}
\acrodef{s.t.}[s.t.]{such that}
\acrodef{MAF}[MAF]{misspecified ambiguity function}
\acrodef{AF}[AF]{ambiguity function}
\acrodef{PMF}[PMF]{probability mass function}
\acrodef{PSD}[PSD]{power spectral density}
\acrodef{CDF}[CDF]{cumulative distribution function}
\acrodef{SNR}[SNR]{signal-to-noise ratio}
\acrodef{SNRs}[SNRs]{signal-to-noise ratios}
\acrodef{iid}[i.i.d.]{independent and identically distributed}
\acrodef{CRB}[CRB]{Cram$\acute{\text{e}}$r-Rao bound}
\acrodef{MCRB}[MCRB]{misspecified CRB}
\acrodef{FIM}[FIM]{Fisher information matrix}
\acrodef{MBCRB}[MBCRB]{misspecified Bayesian Cram$\acute{\text{e}}$r-Rao bound}
\acrodef{MVDR}[MVDR]{minimum-variance distortionless response}
\acrodef{AT-MBCRB}[AT-MBCRB]{asymptotically tight misspecified Bayesian Cram$\acute{\text{e}}$r-Rao bound}
\acrodef{PDF}[PDF]{probability density function}
\acrodef{w.r.t.}[w.r.t.]{with respect to}
\acrodef{AWGN}[AWGN]{additive white Gaussian noise}
\acrodef{ML}[ML]{maximum-likelihood}
\acrodef{DOA}[DOA]{direction-of-arrival}
\acrodef{KLD}[KLD]{Kullback-Leibler divergence}

\input{shortcuts}

\DeclareMathAlphabet\mathbfcal{OMS}{cmsy}{b}{n}

\title{MCRB for Parameter Estimation from One-Bit Quantized and Oversampled Measurements}
\author{
Nadav E. Rosenthal, \emph{Student Member, IEEE}, and Joseph Tabrikian,~\IEEEmembership{Fellow,~IEEE}
 \vspace{-20pt}
\thanks{{This research was partially supported by THE ISRAEL SCIENCE FOUNDATION (grant No. 2493/23).\newline
\indent The authors are with the School of Electrical and Computer Engineering,
Ben-Gurion University of the Negev, Beer-Sheva 84105, Israel (e-mail: rosenthn@post.bgu.ac.il; joseph@bgu.ac.il)
}
}

}
\begin{document}

\maketitle
\nopagebreak

\begin{abstract}
One-bit quantization has garnered significant attention in recent years for various signal processing and communication applications. Estimating model parameters from one-bit quantized data can be challenging, particularly when the quantization process is explicitly accounted for in the estimator. In many cases, the estimator disregards quantization effects, leading to model misspecification. Consequently, estimation errors arise from both quantization and misspecification. Traditional performance bounds, such as the Cramér-Rao bound (CRB), fail to capture the impact of misspecification on estimation performance. To address this limitation, we derive the misspecified CRB (MCRB) for parameter estimation in a quantized data model consisting of a signal component in additive Gaussian noise. We apply this bound to direction-of-arrival estimation using quantized measurements from a sensor array and to frequency estimation with oversampled quantized data. The simulations show that the MCRB is asymptotically achieved by the mean-squared-error of the misspecified maximum-likelihood estimator. Our results demonstrate that, unlike in finely quantized scenarios, oversampling can significantly enhance the estimation performance in the presence of misspecified one-bit quantized measurements.
\end{abstract}

\section{Introduction}
\label{Intro_Section}
Quantization is commonly used in communications and digital signal processing. 
For example, in array signal processing, \ac{ADC} \cite{761034,7258493} convert the analog signal received from the antennas into discrete values, which can be stored and processed by a digital computer. Reducing the number of bits required to represent the collected data decreases the storage and bandwidth requirements, which are essential for data compression and transmission. 

Implementations of one-bit quantization in signal processing and communication systems, which are inexpensive and involve low
energy consumption, have been vastly investigated \cite{272490,mo2017channel,7961157,4558487,7742960,7931630,mollen2016uplink,9413979,9053855,8010806,7837644,8240730}. Recent works have investigated the effects of one-bit quantization on channel estimation \cite{7472499,8314750,708938,5501995}, frequency estimation \cite{824661}, and \ac{DOA} estimation \cite{1039405}. Many system designs have been suggested based on one-bit quantized data in \ac{MIMO} and massive \ac{MIMO} problems \cite{7439790,9253719,8450809,9984699,7472305}.

Quantization introduces signal distortion which leads to system performance degradation \cite{1057548}. In addition, signal processing algorithms that are aware of the quantization become complex and involve higher computational complexity, which may increase the cost as well as higher energy consumption. In \cite{1039405}, a \ac{MVDR} estimator for \ac{DOA} estimation based on reconstructed covariance matrix of one-bit quantized measurements was proposed. However, this estimator collapses for high \ac{SNRs}, because the reconstructed covariance matrix becomes singular. Other estimators based on empirical methods
for evaluating the orthant probabilities have been suggested \cite{1039405}, but they require Monte-Carlo simulations and lack closed-form expressions. Channel estimation in massive \ac{MIMO} systems based on low-resolution is also challenging due to the nonlinear distortion caused by quantization \cite{mo2017channel}. Several channel estimators have been proposed based on deep neural networks \cite{9152072,10145921,balevi2020high,8807322,8918799}.

An appealing and simple approach for processing one-bit quantized data is to implement conventional estimators that ignore the effect of quantization. The question that arises is: what is the expected performance when conventional algorithms that disregard quantization, are applied to one-bit quantized data? In this case, estimation errors are due to both quantization effect and model misspecification. 

The effect of one-bit quantization on parameter estimation can be investigated via estimation performance bounds. Performance bounds provide important tools in signal processing, since they serve as a benchmark for performance evaluation of estimators and are useful for system design. The most common performance bound in the non-Bayesian framework is the \ac{CRB} \cite{CRB1,CRB2}. Its popularity stems from its simplicity and asymptotic attainability. 

Derivation of lower bounds for estimation performance with one-bit quantized measurements has been challenging. For example, an explicit formula for the \ac{CRB} for frequency estimation performance does not exist, and the \ac{CRB} for large-enough \ac{SNR} was investigated \cite{824661}. For the \ac{DOA} estimation problem investigated in \cite{1039405}, the \ac{CRB} was derived explicitly only for a two-sensor array. The \ac{CRB} and Fisher information for channel estimation for oversampled quantized measurements cannot be analytically derived \cite{9085989}, because the derivations require evaluation and explicit formulas for orthant probabilities \cite{10.5555/1695822,9fea1a2a-fee6-329f-a9d3-78a7cb246e6b,19990306}. Thus, approximations for the \ac{CRB} were proposed based on Fisher information lower bounds \cite{9085989,7953006,8445905,Stein_2014}. 

The \ac{CRB} assumes perfect model specification by the estimator, thus it cannot consider the influence of misspecification on the estimation performance. Therefore, it is unable to predict the expected estimation performance of conventional estimators implemented on one-bit quantized data. 

In \cite{White,White_Book,Noam_Tabrikian} the asymptotic properties of the \ac{ML} estimator were investigated under misspecified models where the data samples are statistically independent. 
The \ac{MCRB} was derived in \cite{Vuong86} as an extension of the \ac{CRB} to misspecified scenarios.  Continuing this theory, the effects of model misspecification on the \ac{MSE} of the \ac{ML} in the asymptotic region were studied in \cite{Richmond,Richmond2}. An important overview on misspecified parameter estimation can be found in \cite{Stefano3}.

Many works have investigated the advantages of oversampling quantized data on information rates and channel capacity \cite{9085989,9354165,265506,8444445,6362760,5662127,335948,10.1109/TSP.2024.3356253,8006896}. 
It is well-known that in case of $\infty$-bit quantization, oversampling a bandlimited signal beyond the Nyquist rate does not improve the estimation or signal recovery performance \cite{10.5555/248702}. However, in \cite{5662127,335948}, it was shown that oversampling a bandlimited signal with \ac{AWGN} and one-bit quantization outperforms Nyquist sampling in terms of capacity per unit-cost and achievable rate at high \ac{SNR}. The effect of oversampling one-bit quantized data on estimation performance is unknown. Therefore, a tool for analysis of such problems may be very helpful in signal processing theory and applications. 

In this paper, the effects of model misspecification due to ignoring one-bit quantization are investigated. We derive the \ac{MCRB} and the expected bias of estimators that ignore one-bit quantization. Using the derived \ac{MCRB}, we analyze the expected performance of estimators implemented on one-bit quantized data. Moreover, we use the \ac{MCRB} in order to show how the estimation performance can be improved by oversampling of one-bit quantized data. The computational complexity of the estimation procedure with fine quantization and one-bit measurements are compared in
terms of computational complexity.
We consider the problems of \ac{DOA} and frequency estimation using one-bit quantized measurements, and evaluate the expected bias and the \ac{MCRB}, compared to the misspecified \ac{ML} estimator. 
In the problem of \ac{DOA} estimation, the derived \ac{MCRB} is compared with the \ac{CRB} based on one-bit quantized measurements.
Moreover, the \ac{MCRB} identifies scenarios where misspecified quantization heavily degrades the performance. 
In the problem of frequency estimation, oversampling of the quantized data is shown to improve the estimation performance. A conference version of this work with some preliminary results is expected to appear in \cite{NadavJosephOneBit}. This paper includes several extensions, such as derivation of \ac{MCRB} for colored noise, performance analysis of oversampling for the problem of frequency estimation through the \ac{MCRB}, and computational complexity analysis of the \ac{ML} estimator with misspecified one-bit quantized measurements.

The main contributions of this paper are:
\begin{itemize}
    \item Derivation of the \ac{MCRB} for predicting the expected performance of estimators that ignore one-bit quantization. Its final form is more explicit and easy to evaluate than the \ac{CRB} for one-bit quantized measurements. 
    \item Application of the derived \ac{MCRB} for the problems of \ac{DOA} estimation and frequency estimation with one-bit quantized measurements.
    \item Exploration of oversampling effects on estimation performance with one-bit quantized measurements. It is shown that
    unlike the case of $\infty$-bit quantization, oversampling of one-bit quantized data can improve the expected estimation  performance. 
\end{itemize}

Throughout this paper, the following notations are used. Boldface lowercase and boldface uppercase letters are used to denote vectors and matrices, respectively. Unbold letters of either lowercase or uppercase are used for scalars. Superscripts $^T$, $^H$, and $^*$ stand for transpose, conjugate transpose, and conjugation operations, respectively. The real and imaginary parts of a variable $b\in \mathbb{C}$ are denoted by $\Re(b)$ or $b_{R}$, and $\Im(b)$ or $b_{I}$, respectively.
The gradient of a scalar $b$ \ac{w.r.t.} $\avec \in\mathbb{R}^K$ is a column vector, whose $j$-th element is defined as $\left[\frac{\diff b(\avec)}{\diff \avec}\right]_{j}\triangleq\frac{\partial b(\mathbf{c})}{\partial c_j}\Big|_{\mathbf{c}=\avec}$. 
Given a vector $\mathbf{b}$, its derivative \ac{w.r.t.} $\avec$ is a matrix whose $j,k$-th entry is defined as $\left[\frac{\diff \mathbf{b}}{\diff \avec}\right]_{j,k}\triangleq\frac{\partial b_j(\mathbf{c})}{\partial c_k}\Big|_{\mathbf{c}=\avec}$. The first and second-order derivatives of a function vector $\bvec(c)$ \ac{w.r.t.} the scalar $c$ are denoted by $\Dot{\bvec}(c)\triangleq \frac{d \bvec(c)}{d c}$ and $\Ddot{\bvec}(c)\triangleq \frac{d^2 \bvec(c)}{d c^2}$, respectively. The notation $\Amat\succeq\mathbf{B}$ implies that $\Amat-\mathbf{B}$ is a positive-semidefinite matrix, where $\Amat$ and $\mathbf{B}$ are symmetric matrices of the same size. $\det(\Amat)$ and $tr(\Amat)$ stand for the determinant and trace of the matrix $\Amat$, respectively.  
$\Diag(\avec)$ is a diagonal matrix whose diagonal elements are the entries of the vector $\avec$. Column vectors of size $N$, whose entries are equal to 0 or 1 are denoted by $\zerovec_N$ and $\onevec_N$, respectively, and the identity matrix of size $N\times N$ is denoted by $\Imat_N$. 

This paper is organized as follows. Section \ref{sec: One-Bit Quantized Data Model} presents the model for one-bit quantized measurements. 
In Section \ref{sec: MCRB derivation} the \ac{MCRB} is derived for estimation procedures that ignore the one-bit quantization. The derivation is extended to the case of over-sampled quantized data models. In Section \ref{sec: algorithms}, \ac{ML} estimation based on fine quantization and one-bit measurements are presented and compared in terms of computational cost. Simulation results are presented in Section \ref{sec: Simulations}.
Finally, our conclusions are summarized in Section \ref{sec: Conclusion}. 

\section{One-Bit Quantized Data Model}
\label{sec: One-Bit Quantized Data Model}
Consider the problem of estimating the vector parameter $\thetavec \in \Omega_{\thetavec} \subseteq \mathbb{R}^M$ from a nonlinear model with additive noise. The data model before quantization is given by
\begin{align} \label{eq: x model}
    \xvec = \svec(\thetavec) + \vvec,
\end{align}
where $\vvec \in \mathbb{C}^N
$ is complex Gaussian noise with zero mean and covariance matrix $\Rmat$, and $\svec:  \Omega_{\thetavec} \rightarrow \mathbb{C}^N$
is the signal function vector. The elements of the covariance matrix $\Rmat$ are denoted by
\begin{align}
    \label{eq: R elements}\left[\Rmat\right]_{n,m} \triangleq \left\{ 
    \begin{array}{ll}
    \sigma^2_n, & n=m \\
    \rho_{n,m}, &  n\neq m
    \end{array}
\right.,\quad n,m=1,\dots,N. 
\end{align} 
The real and imaginary parts of the measurements are one-bit quantized:
\begin{align} \label{eq: z_n}
\begin{split}
    z_n &= \text{sign}\left(x_{n,R} \right) + j \cdot \text{sign}\left(x_{n,I} \right),\quad n =1,\dots,N,
\end{split}
\end{align}
where  \begin{align}
    \text{sign}\left(x\right) \triangleq \left\{ 
    \begin{array}{ll}
    +1, & x \geq 0\\
    -1, &  x < 0
    \end{array}
\right..
\end{align} 

The \ac{PMF} of the real and imaginary parts of the quantized measurements (\ref{eq: z_n}) can be derived as follows \cite{8314750}:
\begin{align}\label{eq: z_n probabilites 1}
\begin{split}
    p_{z_{n,R};\thetavec}(z_{n,R};\thetavec) &= Q\left(z_{n,R} q_{n,R}(\thetavec)\right),\quad z_{n,R} \in \{-1,1\},\\
\end{split}
\end{align}
and
\begin{align}\label{eq: z_n probabilites 2}
    \begin{split}
p_{z_{n,I};\thetavec} \left(z_{n,I};\thetavec\right) &= Q\left(z_{n,I} q_{n,I}\left(\thetavec\right)\right),\quad z_{n,I} \in \{-1,1\},
    \end{split}
\end{align}
where 
\begin{align} \label{eq: q_n}
    q_n (\thetavec) \triangleq -\frac{s_n(\thetavec)}{\frac{\sigma_n}{\sqrt{2}}}, \quad n=1,\dots,N,
\end{align}
$s_n(\thetavec)$ is the $n$-th signal function, and 
\begin{align}
    Q(x) \triangleq \frac{1}{\sqrt{2\pi}} \int_{x}^{\infty}e^{-\frac{u^2}{2}}du
\end{align} is the standard Q-function. 
Let the series of events $\{\omega^{(k_n)}_{n}\}^N_{n=1},\quad k_n=0,\dots,3,$ specify the possible values of the elements of $\zvec \triangleq\left[z_1,\dots,z_N\right]^T$: 
\begin{align}
\label{eq: z elements events} \begin{split}
\omega^{(k_n)}_{n} \triangleq & \left\{ 
\begin{array}{ll}
    z_n = -1-j ,& \text{if} \:\: k_n=0\\
    z_n = -1+j ,& \text{if} \:\: k_n=1\\
    z_n = 1-j ,& \text{if} \:\: k_n=2\\
    z_n = 1+j ,& \text{if} \:\: k_n=3\\
    \end{array} 
\right.\\ 
=& \left\{ 
\begin{array}{ll}
    x_{n,R} < 0 \cap x_{n,I}<0,& \text{if} \:\: k_n=0\\
    x_{n,R} < 0 \cap x_{n,I}\geq0 ,& \text{if} \:\: k_n=1\\
    x_{n,R} \geq 0 \cap x_{n,I}<0 ,& \text{if} \:\: k_n=2\\
    x_{n,R} \geq 0 \cap x_{n,I}\geq 0 ,& \text{if} \:\: k_n=3\\
    \end{array} 
\right., n=1,\dots N,
\end{split}
\end{align}
where the second equality in (\ref{eq: z elements events}) is obtained by the quantization operators in (\ref{eq: z_n}).
The \ac{PMF}s of the quantized measurement vector, $\zvec$, are given by using (\ref{eq: z elements events}):
\begin{align}
    \label{eq: PMF zvec}
p_{\zvec;\thetavec}\left(\zvec;\thetavec\right) \triangleq Prob\left(\cap_{n=1}^N \omega^{(k_n)}_{n}\right),\quad \{k_n\}_{n=1}^{N} \in \{0,1,2,3\}.
\end{align}
The set of $4^N$ probabilities defined in (\ref{eq: PMF zvec}) is a set of non-centered orthant probability
functions (multivariate version of the Q-function) of the vector composed of the real and imaginary parts of the elements of $\xvec$ \cite{10.5555/1695822,9fea1a2a-fee6-329f-a9d3-78a7cb246e6b,19990306}. 
Unfortunately, closed-form expressions for the orthant probabilities in (\ref{eq: PMF zvec}) do not exist. Therefore, derivation of a procedure for estimation of $\thetavec$ based on the true \ac{PMF}s in (\ref{eq: PMF zvec}) is complex and impossible in the general setting. 

Instead of using a complex estimation algorithm that accounts for the quantization, 
a simple approach for processing one-bit quantized data is to implement a conventional estimator that ignores the effect of quantization.
The assumed model $f(\zvec;\thetavec)$ ignores the quantization operator and assumes the vector $\vvec$ is complex white Gaussian with covariance matrix $\sigma^2 \Imat_N$. Under this model, the distribution of  $\zvec$ is  given by
\begin{align} \label{eq: z assumed PDF}
   f: \zvec \sim \mathcal{N}^C\left(\svec(\thetavec),\sigma^2 \Imat_N\right).
\end{align}
In Section \ref{sec: MCRB derivation}, we derive the \ac{MCRB} for estimation procedures that ignore one-bit quantization. The derived \ac{MCRB} provides a useful analysis tool which can serve for both performance evaluation and system design.

\section{MCRB derivation} \label{sec: MCRB derivation}
In this section, we consider a basic model which is common in
many signal processing applications, such as frequency or \ac{DOA}
estimation. 
This model is a special case of the model presented in the previous section, where the signal function vector is given by 
\begin{align} \label{eq: DOA signal vector}
    \svec(\thetavec) \triangleq \beta \avec(\varphi), 
\end{align} 
where $\thetavec \triangleq [\varphi, \beta_R, \beta_I]^T$, $\varphi \in \Omega_{\varphi} \subseteq \mathbb{R}$ is a deterministic and unknown parameter-of-interest, and $\beta \in \mathbb{C}$ is deterministic unknown complex amplitude.
The function $\avec: \Omega_\varphi \rightarrow \mathbb{C}^N$ is assumed to be known and satisfy $\|\avec(\varphi)\|^2=1 $ and $\Dot{\avec}^H(\varphi) \avec(\varphi)=0, \quad \forall \varphi \in \Omega_\varphi$.
We ought to evaluate the quantization impact on the expected estimation performance of conventional estimators that disregard the quantization, in terms of bias and expected covariance \cite{Richmond}. \subsection{Estimation Bias and Pseudo-True Parameter}
For an estimator $\hat{\thetavec}(\zvec)$ the mean bias is defined by:
\begin{align} \label{eq: estimation bias}
\bvec\left(\hat{\thetavec},\thetavec\right) \triangleq \E \left[\hat{\thetavec}(\zvec)\right] - \thetavec,   
\end{align}
while the property of misspecified (MS) unbiasedness \cite{Stefano3} is defined by $\E\left[\hat{\thetavec}(\zvec)\right] = \thetavec_0$,
where $\thetavec_0$ is the pseudo-true parameter vector given by \cite{Vuong86}
\begin{align} \label{eq: pseudo true parameter}
\begin{split}
    \thetavec_0 &\triangleq \arg \max_{\thetavec'\in \Omega_\thetavec} \E \left [\log f(\zvec;\thetavec')  \right ].
\end{split}
\end{align}
The expectations in (\ref{eq: estimation bias}) and (\ref{eq: pseudo true parameter}), as well as throughout the paper, are taken \ac{w.r.t.} the true \ac{PMF}, and true parameter $\thetavec$. For cases that the true and assumed distributions are continuous, the pseudo-true parameter is alternatively defined by the parameter that minimizes the \ac{KLD} between the true and assumed distributions \cite{Stefano}. The misspecified \ac{ML} estimator is known to be asymptotically MS-unbiased.  
In this subsection, the estimation mean-bias in (\ref{eq: estimation bias}) is investigated for estimators $\hat{\thetavec}(\zvec)$ that are MS-unbiased and satisfy
\begin{align} \label{eq: MS-unbiased bias}
    \bvec\left(\hat{\thetavec},\thetavec\right) = \thetavec_0 - \thetavec.
\end{align}
The assumed log-likelihood at the parameter vector $\thetavec'$ is given by using (\ref{eq: z assumed PDF}) and (\ref{eq: DOA signal vector}):
\begin{align}\label{eq: assumed log-likelihood}
    \log f\left(\zvec;\thetavec'\right) = -N\log\left(\pi \sigma^2\right) - \frac{\|\zvec-\beta' \avec(\varphi')\|^2}{\sigma^2},
\end{align}
where $\thetavec' = \left[\varphi',\beta'_R,\beta'_I\right]^T$.
By substituting (\ref{eq: assumed log-likelihood}) into the \ac{r.h.s.} of (\ref{eq: pseudo true parameter}), removing the terms that are independent of $\thetavec'$
and using $\|\avec(\varphi')\|^2=1$, one obtains
\begin{align} \label{eq: pseudo true argmax 1}
    \thetavec_0 = \arg \max_{\thetavec'\in \Omega_\thetavec} \left[2\Re \left(\beta' \muvec^H(\thetavec)  \avec(\varphi')\right)- \left|\beta'\right|^2 \right], 
\end{align}
where $\muvec(\thetavec) \triangleq \E \left[\zvec\right]$ is derived in (\ref{eq: mu R n}) and (\ref{eq: mu I n}) in Appendix A and depends on both $\varphi$ and $\beta$. Maximizing the expression in (\ref{eq: pseudo true argmax 1}) \ac{w.r.t.} $\thetavec'$ yields 
the pseudo-true parameter vector
\begin{align}
    \label{eq: pseudo-true theta}
    \thetavec_0 = \left[\varphi_0 ,\beta_{0,R},\beta_{0,I}\right]^T,
\end{align}
where the pseudo-true parameter-of-interest is 
\begin{align}\label{eq: varphi_0 maximized}
    \varphi_0 = \arg \max_{\varphi'} \left| MAF(\varphi',\thetavec) \right|,
\end{align}
in which the \ac{MAF} is defined by 
\begin{align}\label{eq: MAF}
    MAF(\varphi',\thetavec) \triangleq \frac{\avec^H(\varphi') \muvec(\thetavec)}{\max\limits_{\varphi'} \left| \avec^H(\varphi') \muvec(\thetavec)\right|},
\end{align}
and the pseudo-true complex amplitude is 
\begin{align}
    \beta_0 = \avec^H(\varphi_0) \muvec(\thetavec).
\end{align}
Note that the pseudo-true parameter-of-interest in (\ref{eq: varphi_0 maximized}) is derived by the parameter that maximizes the absolute of \ac{MAF} in (\ref{eq: MAF}), which is the correlation between the quantized measurements expectation, $\muvec(\thetavec)$, and the function, $\avec(\cdot)$.
Examining the 
\ac{MAF}, we note that it depends on $\thetavec$ through $\muvec(\thetavec)$. Specifically, it depends on the true complex amplitude $\beta$ by a nonlinear function. 
For the unquantized model in (\ref{eq: x model}), the \ac{AF} can be defined by $AF(\varphi',\varphi) \triangleq \frac{\avec^H(\varphi') \avec(\varphi)}{\|\avec(\varphi')\| \|\avec(\varphi)\|}$ (see \cite{10154126}),
which obtains its maximum at $\varphi'=\varphi$. Notice that the \ac{AF} is independent on $\beta$.
By substitution of (\ref{eq: varphi_0 maximized}) in the \ac{r.h.s.} of (\ref{eq: MS-unbiased bias}) and evaluating its first element, one obtains the estimation bias of the parameter-of-interest, $\varphi$,
\begin{align}
    \label{eq: DOA estimation bias}    b_1\left(\hat{\thetavec},\thetavec\right)  =  \arg \max_{\varphi'} \left| MAF(\varphi',\thetavec)\right| - \varphi.
\end{align}
\subsection{\ac{MSE} Performance and \ac{MCRB}}
In this subsection, the \ac{MCRB} is derived to investigate the expected \ac{MSE} performance for a MS-unbiased estimator $\hat{\thetavec}(\zvec)$, derived from the assumed model in (\ref{eq: z assumed PDF}). 
The MCRB for estimation of $\thetavec$ from the quantized measurements $\zvec$ is given by \cite{Stefano3}
\begin{align} \label{eq: MCRB}
    \mathbf{MCRB}(\thetavec) \triangleq \Amat^{-1}_{\thetavec_0} \Bmat_{\thetavec_0} \Amat^{-1}_{\thetavec_0}, 
\end{align}
where the matrices $\Bmat_{\thetavec_0}$ and $\Amat_{\thetavec_0}$ are defined by 
\begin{align} \label{eq: MCRB matrices 1}
    \Bmat_{\thetavec_0} \triangleq \E \left[\dvec(\zvec,\thetavec_0) \dvec^T(\zvec,\thetavec_0)\right],
\end{align}
\begin{align} \label{eq: MCRB matrices 2}
    \Amat_{\thetavec_0} \triangleq \E \left[\Hmat \left(\zvec,\thetavec_0\right)\right],
\end{align}
and
the gradient and Hessian of the assumed log-likelihood \ac{w.r.t.} $\thetavec$ at $\thetavec_0$ are given by
\begin{align}
    \label{eq: log-likelihood gradient}
    \dvec(\zvec,\thetavec_0) \triangleq \frac{\partial \log f\left(\zvec;\thetavec_0\right)}{\partial \thetavec}
\end{align}
and 
\begin{align}
    \label{eq: log-likelihood Hessian}
    \Hmat \left(\zvec,\thetavec_0\right) \triangleq \frac{\partial^2 \log f\left(\zvec;\thetavec_0\right)}{\partial \thetavec \partial \thetavec^T},
\end{align}
respectively.
By substitution of (\ref{eq: assumed log-likelihood}) at (\ref{eq: pseudo-true theta}) in (\ref{eq: log-likelihood gradient}) and (\ref{eq: log-likelihood Hessian}), followed by 
few calculations, we obtain
\begin{align}
    \label{eq: assumed log-likelihood gradient}
    \begin{split}
\dvec(\zvec,\thetavec_0) &= \frac{2}{\sigma^2}\begin{bmatrix}
d_1 \\
d_{2,R}  \\
d_{2,I}\\
\end{bmatrix}, 
    \end{split}
\end{align}
and
\begin{align}\label{eq: assumed log-likelihood Hessian}
    \begin{split}
\Hmat \left(\zvec,\thetavec_0\right) &= \frac{2}{\sigma^2} \begin{bmatrix} 
    h_1& h_{2,R}& h_{2,I} \\
    h_{2,R}& -1& 0\\
    h_{2,I}& 0& -1\\
    \end{bmatrix},\\
    \end{split}
\end{align}
where $d_1 = \Re\left( \zvec^H \Dot{\avec}(\varphi_0) \beta_0\right)$, $d_2 = \avec^H(\varphi_0) \zvec - \beta_0$, $h_1 \triangleq \Re\left(\zvec^H \Ddot{\avec}(\varphi_0) \beta_0\right)$ and $h_2 \triangleq \Dot{\avec}^H(\varphi_0) \zvec$. By substituting (\ref{eq: assumed log-likelihood gradient}) and (\ref{eq: assumed log-likelihood Hessian}) into (\ref{eq: MCRB matrices 1}) and (\ref{eq: MCRB matrices 2}), respectively, and using (\ref{eq: mu n}) from Appendix A, we obtain after a few lines of algebra,
    \begin{align}
        \begin{split} \label{eq: B theta}
          \Bmat_{\thetavec_0}   &= \frac{2}{\sigma^4} \begin{bmatrix} 
    L_1 & L_{2,R}& L_{2,I} \\
    L_{2,R}& L_5& L_4\\
    L_{2,I}& L_4& L_6\\
    \end{bmatrix} 
        \end{split}
    \end{align} 
    and
        \begin{align} \label{eq: A theta}
        \begin{split}
    \Amat_{\thetavec_0} &= \frac{2}{\sigma^2} \begin{bmatrix} 
    J_1 & J_{2,R}& J_{2,I} \\
    J_{2,R}& -1& 0\\
    J_{2,I}& 0& -1\\
    \end{bmatrix} ,\end{split}
    \end{align}
where the elements of (\ref{eq: B theta}) and (\ref{eq: A theta}) are 
\begin{align} \label{eq: A and B elements}
 \begin{split}
J_1 \triangleq & \E\left[h_1\right] =\Re\left(\muvec^H(\thetavec) \Ddot{\avec}(\varphi_0) \beta_0\right), \\
J_2 \triangleq &\E \left[h_2\right] = \Dot{\avec}^H(\varphi_0) \muvec(\thetavec),\\
L_1 \triangleq &\Re\left(\beta_0^2\Dot{\avec}^T(\varphi_0) \Pmat^{*} \Dot{\avec}(\varphi_0)\right) + \left|\beta_0\right|^2\Dot{\avec}^H(\varphi_0) \Mmat \Dot{\avec}(\varphi_0),\\
L_2 \triangleq&  \beta^*_0\Dot{\avec}^H(\varphi_0) \Pmat \avec^*(\varphi_0) \\
 &+ \beta_0 \left(\avec^H(\varphi_0) \Mmat \Dot{\avec}(\varphi_0) -2 \Re \left(J_2 \beta_0^*\right)\right),\\
L_3 \triangleq & \avec^H(\varphi_0) \Pmat \avec^*(\varphi_0),\\
L_4 \triangleq& 
L_{3,I} -2 \beta_{0,R}  \beta_{0,I},\\
    L_5 \triangleq &L_{3,R} + \avec^H(\varphi_0) \Mmat \avec(\varphi_0)-2\left(\beta_{0,R}\right)^2,\\
L_6 \triangleq & -L_{3,R} + \avec^H(\varphi_0) \Mmat \avec(\varphi_0)-2\left(\beta_{0,I}\right)^2,
    \end{split}
\end{align}
and the second-order moment and second-order pseudo-moment matrices of the quantized measurements vector, $\zvec$, are denoted by
$\Mmat \triangleq \E\left[\zvec \zvec^H\right]$ and $\Pmat \triangleq \E\left[\zvec \zvec^T\right]$, respectively.  Substitution of (\ref{eq: A theta}) and (\ref{eq: B theta}) in (\ref{eq: MCRB}) followed by few simple calculations form the \ac{MCRB} for the parameter-of-interest, $\varphi$, with one-bit quantized data
\begin{align}
    \label{eq: DOA MCRB}
    \begin{split}
&\left[\mathbf{MCRB}(\thetavec)\right]_{1,1} =\\
&\frac{L_1 +2 \Re\left(J_2 L^{*}_2\right)+ L_6 J^{2}_{2,I} + L_5 J^{2}_{2,R} + 2 L_4 J_{2,I} J_{2,R}}{2\left(\left|J_2\right|^2 + J_1\right)^2 }.
    \end{split}
\end{align}
Finally, the corresponding bound on the \ac{MSE} of the estimator $\hat{\varphi}(\zvec)$ of $\varphi$ is given by using (\ref{eq: DOA estimation bias}) and (\ref{eq: DOA MCRB}) \cite{Stefano3}
\begin{align}
    \label{eq: MSE bound}
    \begin{split}
MSE\left(\hat{\varphi}(\zvec)\right) \triangleq& \E \left[\left(\hat{\varphi}(\zvec)-\varphi\right)^2\right] \geq \\
&\left[\mathbf{MCRB}(\thetavec)\right]_{1,1} +  b^2_1\left(\hat{\thetavec},\thetavec\right).    
    \end{split}
\end{align}
For evaluating the \ac{MSE} bound in (\ref{eq: MSE bound}), explicit expressions for the matrices $\Mmat$ and $\Pmat$ are required. In Subsection \ref{sec: MCRB for Oversampling quantized data} we derive the \ac{MCRB} for oversampled one-bit quantized measurements. Oversampling in a frequency larger than the Nyquist rate results in colored noise. Thus, in the following subsections, we analyze the elements of the matrices $\Mmat$ and $\Pmat$ for the general case of colored noise, $\vvec$, with covariance matrix $\Rmat$, and the special case where the noise, $\vvec$, is \ac{AWGN}.

\subsubsection{\ac{MCRB} for One-Bit Quantized Data in Colored Noise}

\quad
The elements $\left[\Mmat\right]_{i,l}$ and $\left[\Pmat\right]_{i,l}$ are given by
\begin{align}
    \label{eq: M def}\left[\Mmat\right]_{i,l} \triangleq \E 
    \left[z_i z^*_l\right],\quad i,l=1,\dots,N 
\end{align}
and 

\begin{align}
     \label{eq: P def}\left[\Pmat\right]_{i,l} \triangleq \E 
    \left[z_i z_l\right],\quad i,l=1,\dots,N,  
\end{align} respectively.
Evaluating (\ref{eq: M def}) and (\ref{eq: P def}) for $i=l$ while using the statistics of quantized random variables in (\ref{eq: mu n})-(\ref{eq: z abs and squared}) in Appendix A, yields 
\begin{align}\label{eq: second moment and pseudo-second moment 1}
\begin{split}
\left[\Mmat\right]_{i,i} &= \E\left[\left|z_i\right|^2\right]=2
\end{split}    
\end{align}
and 
\begin{align}
\begin{split}
\label{eq: second moment and pseudo-second moment 2}\left[\Pmat\right]_{i,i} &= \E\left[z^2_i\right]= 2j \cdot \mu_{i,R}(\thetavec) \mu_{i,I}(\thetavec).
    \end{split}
\end{align}
Evaluating the elements (\ref{eq: M def}) and (\ref{eq: P def}) for $i \neq l$ gives 
\begin{align}
\label{eq: expectation z_i z_l*}
\begin{split}
\left[\Mmat\right]_{i,l} &=  \sum_{
\small{\begin{matrix}
z_i,z_l\in\{\pm 1 \pm j\}, \\
   k_i,k_l=0,\dots,3
\end{matrix}
}
} P^{(k_i,k_l)}_{i,l}      z_i z_l^*,
 \end{split}
\end{align}
and 
\begin{align}
\label{eq: expectation z_i z_l}
\begin{split}
\left[\Pmat\right]_{i,l} & = \sum_{
\small{\begin{matrix}
z_i,z_l\in\{\pm 1 \pm j\}, \\
   k_i,k_l=0,\dots,3
\end{matrix}
}
}P^{(k_i,k_l)}_{i,l}      z_i z_l,
 \end{split}
\end{align}
where the \ac{PMF}s of $\zvec_{i,l}\triangleq \left[z_i,z_l\right]^T$ are denoted by $P^{(k_i,k_l)}_{i,l}, \quad k_i,k_l=0,\dots,3$ and obtained from the quantization rule (\ref{eq: z_n}) and the events described in (\ref{eq: z elements events}) for $z_i$ and $z_l$:
\begin{align}
\label{eq: z_i,l PMF}
\begin{split}
P^{(k_i,k_l)}_{i,l} = Prob\left(\omega^{(k_i)}_{i} \cap \omega^{(k_l)}_{l} \right),\quad k_i,k_l&=0,\dots,3,\\
i,l &=1,\dots,N.  
\end{split}
\end{align}
Let the vector $\xvec_{i,l}$ be defined as
\begin{align} \label{eq: x bar i,l}
\xvec_{i,l} \triangleq \left[x_{i,R},x_{i,I},x_{l,R},x_{l,I}\right]^T,\quad  i,l=1,\dots,N.
\end{align}
According to the model in (\ref{eq: x model}), $\xvec_{i,l} \sim \mathcal{N}\left(\uvec_{i,l},\Cmat_{i,l} \right)$, where
\begin{align} \label{eq: x bar i,l mean}
\begin{split}
    \uvec_{i,l} &\triangleq  \left[s_{i,R}(\thetavec),s_{i,I}(\thetavec),s_{l,R}(\thetavec),s_{l,I}(\thetavec)\right]^T
\end{split}
\end{align}
and
\begin{align} \label{eq: C i,l final}
    \begin{split}
\Cmat_{i,l} \triangleq& \frac{1}{2} \begin{bmatrix}
     \sigma^2_i & 0 &\rho_{i,l} & 0\\
     0 & \sigma^2_i &0 &\rho_{i,l} \\
      \rho_{i,l}& 0 & \sigma^2_l& 0\\
      0& \rho_{i,l} &0 &\sigma^2_l \\
\end{bmatrix}
    \end{split}
\end{align}
are the expectation vector and the covariance matrix of $\xvec_{i,l}$, respectively, where (\ref{eq: C i,l final}) is obtained from properties of the circularly symmetric complex Gaussian distribution with covariance matrix elements in (\ref{eq: R elements}).
The set of 16 probabilities defined in (\ref{eq: z_i,l PMF}) is the quadrivariate set of non-centered orthant probability
functions of $\xvec_{i,l}$ \cite{10.5555/1695822,ef984e50-3519-38ce-adbf-ad8274fdedba,9b2f3c1f-6cc8-3a07-8891-81ab55896b08,9fea1a2a-fee6-329f-a9d3-78a7cb246e6b,19990306}. Evaluation of the set in (\ref{eq: z_i,l PMF}) requires computing the four dimensional integrals 
\begin{align}\label{eq: P i,l integral}
\begin{split}
    P^{(k_i,k_l)}_{i,l} = \int_{B^{(k_i,k_l)}_{i,l}} \phi\left(\xvec_{i,l},\uvec_{i,l},\Cmat_{i,l}\right)&d\xvec_{i,l},\\ 
    k_i,k_l =& 0,\dots,3,
\end{split}
\end{align}
where $\phi\left(\xvec,\uvec,\Cmat\right)$ is the quadrivariate Gaussian \ac{PDF} and the integration region $B^{(k_i,k_l)}_{i,l}$ is defined by 
\begin{align} \label{eq: B i,l}
   B^{(k_i,k_l)}_{i,l} \triangleq& \{\xvec_{i,l}\in \mathbb{R}^4|\omega^{(k_i)}_{i} \cap \omega^{(k_l)}_{l}  \} .
\end{align} 
Although there is no closed-form expression for the integral in the \ac{r.h.s.} of (\ref{eq: P i,l integral}), numerical integration methods like Quasi Monte-Carlo methods \cite{10.5555/1695822,2df842b9-0d57-33ab-9a87-b36efd47c710} can be
applied to approximate it. Substitution of the approximations of (\ref{eq: P i,l integral}) in (\ref{eq: expectation z_i z_l*}) and (\ref{eq: expectation z_i z_l}) gives the elements $\left[\Mmat\right]_{i,l}$ and $\left[\Pmat\right]_{i,l}$ for $i\neq l,\quad i,l=1,\dots,N$, respectively.

\subsubsection{\ac{MCRB} for One-Bit Quantized Data in \ac{AWGN}}

\quad \quad \quad\quad
For the special case of \ac{AWGN} where $\Rmat=\sigma^2\Imat_N$, the covariance matrix in (\ref{eq: C i,l final}) resides to
\begin{align} \label{eq: C i,l uncorrelated}
\Cmat_{i,l}=\frac{1}{2} \sigma^2 \Imat_4, \quad i,l=1,\dots,N.
\end{align}
Therefore, the elements of $\xvec_{i,l}$ for $i \neq l$ are statistically independent, since $\xvec_{i,l}$ is Gaussian, and the probability in (\ref{eq: P i,l integral}) can be expressed by the product of marginal probabilities of $\xvec_{i,l}$. After few lines of algebra, it can be shown that 
\begin{align} \label{eq: P i,l uncorrelated}
\begin{split}
    P^{(k_i,k_l)}_{i,l} 
    =p_{z_{i,R};\thetavec}(z_{i,R};\thetavec) &p_{z_{i,I};\thetavec}(z_{i,I};\thetavec)\\ \times p_{z_{l,R};\thetavec}(z_{l,R};\thetavec)&p_{z_{l,I};\thetavec}(z_{l,I};\thetavec), \\
    z_i,z_l &\in \{\pm 1 \pm j\},
    \end{split}
\end{align}
By substituting (\ref{eq: P i,l uncorrelated}) into (\ref{eq: expectation z_i z_l*}) and (\ref{eq: expectation z_i z_l}) one obtains
\begin{align} \label{eq: Mmat i neq l uncorrelated}
    \begin{split}
\left[\Mmat\right]_{i,l} =&  \sum_{
\small{\begin{matrix}
z_i\in\{\pm 1 \pm j\}
\end{matrix}
}
} p_{z_{i,R};\thetavec}(z_{i,R};\thetavec) p_{z_{i,I};\thetavec}(z_{i,I};\thetavec)   z_i \\
&\times \sum_{
\small{\begin{matrix}
z_l\in\{\pm 1 \pm j\}
\end{matrix}
}
} p_{z_{l,R};\thetavec}(z_{l,R};\thetavec)p_{z_{l,I};\thetavec}(z_{l,I};\thetavec)      z_l^*\\
=& \E\left[z_i\right] \E\left[z^*_l\right]= \mu_i(\thetavec) \mu^*_l(\thetavec), \quad i\neq l,
 \end{split}
\end{align}
and
\begin{align}
    \label{eq: Pmat i neq l uncorrelated}
    \begin{split}
\left[\Pmat\right]_{i,l} 
=& \E\left[z_i\right] \E\left[z_l\right] = \mu_i(\thetavec) \mu_l(\thetavec),\quad i\neq l.
 \end{split}
\end{align}
The diagonal elements of $\Mmat$ and $\Pmat$ are given by (\ref{eq: second moment and pseudo-second moment 1}) and (\ref{eq: second moment and pseudo-second moment 2}), respectively.
\subsection{MCRB for Oversampled Quantized Data}\label{sec: MCRB for Oversampling quantized data}
The derived \ac{MCRB} can be used not only for performance evaluation in the case of one-bit quantization, but also to evaluated and analyze potential performance improvement through oversampling. Nyquist rate sampling of a deterministic band-limited signal with \ac{AWGN} and low-pass filter, followed by one-bit quantization is expected to produce statistically independent observations. However, statistically independence is not preserved in the case of oversampling. 
An overview on sampling continuous-time signals can be found in \cite{10.5555/248702}. The covariance matrix elements of a band-limited continuous time signal in \ac{AWGN} are \cite{7953006}
\begin{align} \label{eq: covarinace sinc}
    \left[\Rmat\right]_{n,m} = \sigma^2 \cdot \text{sinc}\left(2B T_s \left(n-m\right) \right), \quad n,m=1,\dots,N,
\end{align}
where $T_s$ and $f_s=\frac{1}{T_s}$ are the sampling 
time interval and frequency, respectively,
$B$ is the bandwidth of the signal, and 
\begin{align}
    \text{sinc}(x) \triangleq \frac{\sin(\pi x)}{\pi x}.
\end{align} 

If the sampling frequency is exactly equal to the Nyquist rate $f_s=2B$, then (\ref{eq: covarinace sinc}) resides to $\Rmat=\sigma^2 \Imat_N$ and the noise vector $\vvec$ is white. In order to obtain the \ac{MCRB} for parameter estimation with oversampled quantized measurements where $f_s > 2B$, the elements of the matrices $\Mmat$ and $\Pmat$ in (\ref{eq: M def}) and (\ref{eq: P def}) should be evaluated with the covariance matrix in (\ref{eq: covarinace sinc}). Substitution of the resulting $\Mmat$ and $\Pmat$ matrices in (\ref{eq: A and B elements}), followed by their substitution in (\ref{eq: DOA MCRB}) and (\ref{eq: MSE bound}) gives the \ac{MCRB} and \ac{MSE} bound for estimation with oversampled quantized measurements, respectively.

In Section \ref{sec: Simulations}, we evaluate the MSE bounds for one-bit quantized data and show that oversampling may partially compensate for the lost information due to quantization. A similar trade-off exists in terms of computational complexity. While processing of one-bit quantized data is simpler, increasing the sampling rate may involve higher amount of computations. The question that can be raised is whether the computational complexity of the conventional estimator with quantized measurements is reduced compared to the conventional estimator with finite number of bits per measurement.
In the next section, the estimation procedures for the models in (\ref{eq: x model}) and (\ref{eq: z_n}) are presented, and their computational costs are discussed.

\section{Computational Complexity of One-Bit Quantized and Oversampled Measurements}
\label{sec: algorithms}
In this section, we analyze the computational complexity of the misspecified \ac{ML} estimator based on one-bit and fine quantization measurements in terms of the number of real operations: multiplication, addition, or subtraction, $C_{op}$. The algorithm based on one-bit measurements is shown to be more efficient than the fine quantization algorithm, in terms of computational complexity. Moreover, the influence of oversampling on the computational complexity is presented. 

The \ac{ML} estimator under the assumed model
in (\ref{eq: z assumed PDF}) with non-quantized measurements
is given by \cite{van2004optimum}
\begin{align} \label{eq: MML x}
    \hat{\varphi}(\xvec) \triangleq \arg \max_{\varphi'} \left| c\left(\avec(\varphi'),\xvec\right) \right|^2,
\end{align} 
where the correlator $c: \mathbb{C}^N \times \mathbb{C}^N \rightarrow \mathbb{C} $ is defined by
\begin{align} \label{eq: correlator}
c\left(\hvec,\uvec\right) \triangleq \hvec^H \uvec.
\end{align}
The same estimator is used in the quantized case where the input is substituted by $\zvec$: $\hat{\varphi}(\zvec)$. 
\begin{prop} 

Assume the misspecified \ac{ML} estimators based on fine quantization and one-bit quantized measurements are evaluated with grid parameter $\varphi'$ and $K$ possible values of $\varphi'$. 
Let $C_{\zvec}$ and $C_{\xvec}$ define the computational complexity of the estimator with one-bit and fine quantization, respectively. Then, the reduction in the computational complexity for using one-bit measurements is 
\begin{align} \label{eq: CU approx}
    \lim_{N\rightarrow \infty} \frac{C_{\zvec}}{C_{\xvec}} = \frac{1}{4}
\end{align}
\end{prop}
\begin{proof}
    The proof appears in Appendix B.
\end{proof}

For the case of over-sampled quantized measurements $\zvec$, the computational complexity increases by the oversampling factor $U\triangleq \frac{f_s}{2B}$.
Thus, the computational complexity using quantized measurements with oversampling factor $U$, increases by a factor of $\frac{U}{4}$. Accordingly, processing of over-sampled quantized data with $U<4$ is computationally more efficient than processing fine-quantized data at Nyquist rate. 

\section{Simulations} \label{sec: Simulations}
In this section, the expected performance of misspecified one-bit quantized measurements in two signal processing applications is investigated via the \ac{MCRB} and mean-bias. In Subsection \ref{sec: DOA estimation} the problem of \ac{DOA} estimation using a sensor array is considered, and in Subsection \ref{sec: frequency estimation} the problem of frequency estimation is addressed. The \ac{MCRB} is used to evaluate the expected performance for oversampled one-bit quantized data. 

\subsection{\ac{DOA} Estimation}
\label{sec: DOA estimation}
In the problem of \ac{DOA} estimation of a narrowband far-field source, the measurement model before quantization can be described by (\ref{eq: x model}) and (\ref{eq: DOA signal vector}), where $\avec(\varphi)$ is the steering vector of direction $\varphi \in \Omega_{\varphi} \subseteq (-\frac{\pi}{2},\frac{\pi}{2})$ and $\beta \in \mathbb{C}$ is deterministic unknown complex amplitude. 
Consider a \ac{ULA} of $N=16$ sensors with half a wavelength inter-element spacing. Thus, the $n$-th element of the steering vector $\avec(\varphi)$ is given by
\begin{align} \label{eq: DOA a vec elements}  a_n(\varphi) = \frac{1}{\sqrt{N}}\exp\left[j\pi \left(n-\frac{N+1}{2}\right)\sin \varphi \right].
\end{align}
According to the discussion in Section \ref{sec: MCRB for Oversampling quantized data}, the noise vector $\vvec$ is white Gaussian.
The \ac{SNR} is defined by $SNR \triangleq \frac{\left|\beta\right|^2}{\sigma^2}$. Figs. \ref{fig 1 AF} and \ref{fig 2 MAF} depict the \ac{AF} and \ac{MAF}, respectively, versus spatial frequency  $u = \sin \varphi$ for real $\beta$, and $SNR=30$dB. 
The \ac{MAF} that considers model misspecification due to one-bit quantization presents several distinguished properties:
\begin{itemize}
    \item Distorted mainlobes in width and shape, compared to the \ac{AF}.
    \item Sidelobe distance from the mainlobe and their heights depend on $\sin \varphi$.
    \item Three large-valued lobe patterns that cross $\varphi'$ grid and  intersect with the mainlobe.
\end{itemize}
\begin{figure}[htbp]
    \centering
    \begin{subfigure}[b]{0.24\textwidth}
        \centering
\includegraphics[width=\textwidth]{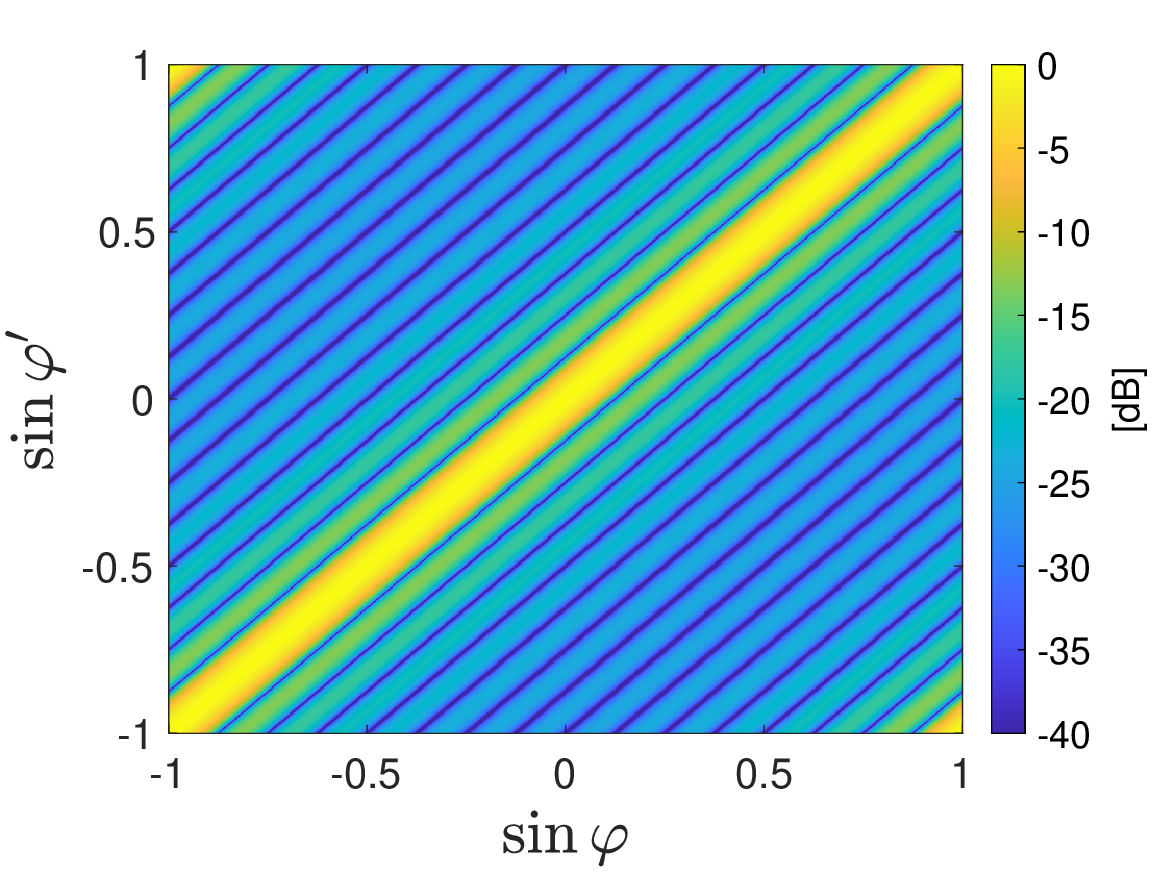}
\caption{$\left|AF\right|$}
\label{fig 1 AF}
\end{subfigure}%
\begin{subfigure}[b]{0.24\textwidth}
    \centering
\includegraphics[width=\textwidth]{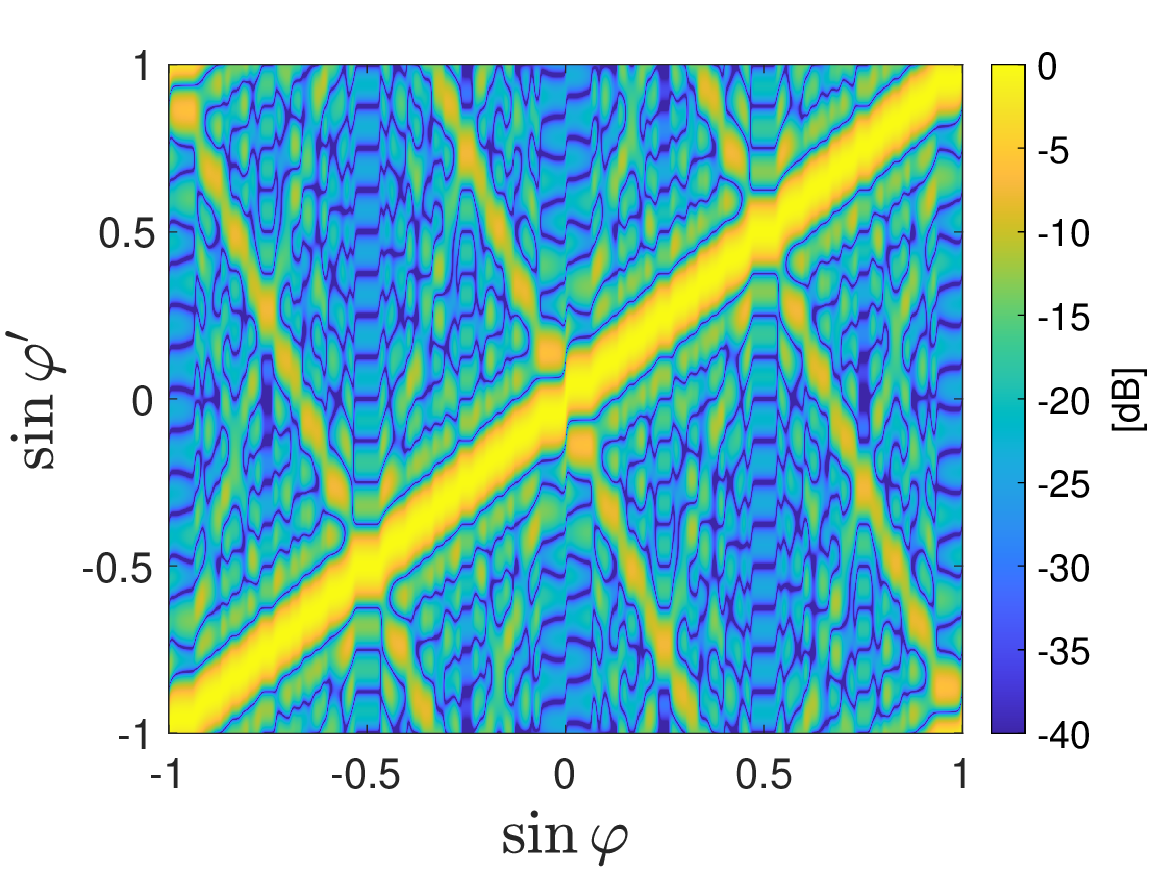} \caption{$\left|MAF\right|$}
\label{fig 2 MAF}
\end{subfigure}
\caption{Comparison between perfectly specified and misspecified ambiguity functions for \ac{DOA} estimation with \ac{ULA} of $N=16$ sensors, complex amplitude phase $\angle{\beta}=0^\circ$, and $SNR=30$dB.} \label{fig:mainfigure}
\end{figure}
Fig. \ref{fig 3 estimation bias} depicts the absolute estimation bias given in (\ref{eq: DOA estimation bias}) versus \ac{DOA} for several \ac{SNR}'s. For $SNR=10$dB the estimation bias is relatively small and it increases as the \ac{SNR} grows. For example, large estimation bias of $9^\circ$ and $1.8^\circ$ at $SNR=30$dB correspond to distorted mainlobes and intersection points in Fig. \ref{fig 2 MAF}. In those \ac{DOA}s the pseudo-true \ac{DOA}, which is the parameter that maximizes the \ac{MAF} along $\sin \varphi'$ grid, is influenced by the distorted mainlobes.
\begin{figure}[htbp]
\centerline{\includegraphics[width=0.5\textwidth]{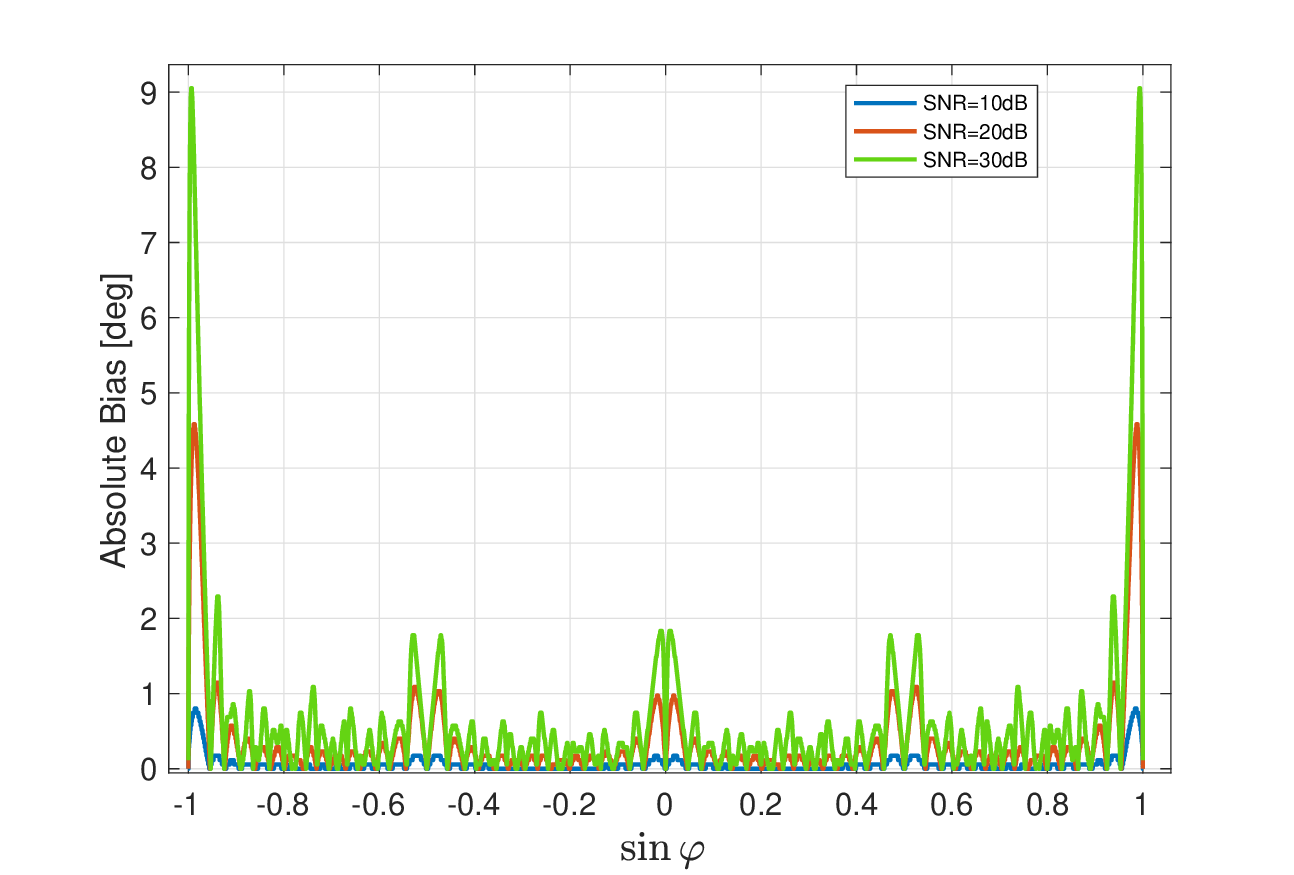}}
\caption{Absolute estimation bias versus \ac{DOA} for several \ac{SNR}'s with \ac{ULA} of $N=16$ sensors, and complex amplitude phase $\angle{\beta}=0^\circ$.}
\label{fig 3 estimation bias}
\end{figure}
The \ac{RMSE} performance of the misspecified \ac{ML} estimator $\hat{\varphi}(\zvec)$ is presented in Fig. \ref{fig 4 MCRB} along with the \ac{MSE} bound (\ref{eq: MSE bound}) (denoted as \ac{MCRB}) for \ac{DOA} $\varphi=0^\circ$, and several complex amplitude phases $\angle{\beta}$ versus \ac{SNR}. 
Moreover, the \ac{CRB} for the quantized model (\ref{eq: z_n}) is derived and presented similarly to \cite{824661,8314750} to obtain
\begin{align} \label{eq: CRB}
    \begin{split}
    CRB(\varphi) =\left[\Jmat^{-1}(\thetavec)\right]_{1,1},
    \end{split}
\end{align}
where the \ac{FIM} is given by
\begin{align}
    \begin{split}
    \Jmat(\thetavec) &= \sum_{n=1}^N \left[\psi \left(q_{n,R} \left(\thetavec\right)\right) \frac{d s_{n,R}(\thetavec)}{d \thetavec} \left(\frac{d s_{n,R}(\thetavec)}{d \thetavec}\right)^T \right.\\
        &\left.\quad\quad\,\,\;+ \psi \left(q_{n,I} \left(\thetavec\right)\right) \frac{d s_{n,I}(\thetavec)}{d \thetavec} \left(\frac{d s_{n,I}(\thetavec)}{d \thetavec}\right)^T\right],
    \end{split}
\end{align} 
in which the function $\psi(\cdot)$ is defined by
\begin{align}
\psi(q) \triangleq \frac{\exp\left[-q^2\right]}{\sigma^2 \pi \left[Q\left(q\right)-Q^2\left(q\right)\right]}.
\end{align}
\begin{figure}[htbp]
\centerline{\includegraphics[width=0.5\textwidth]{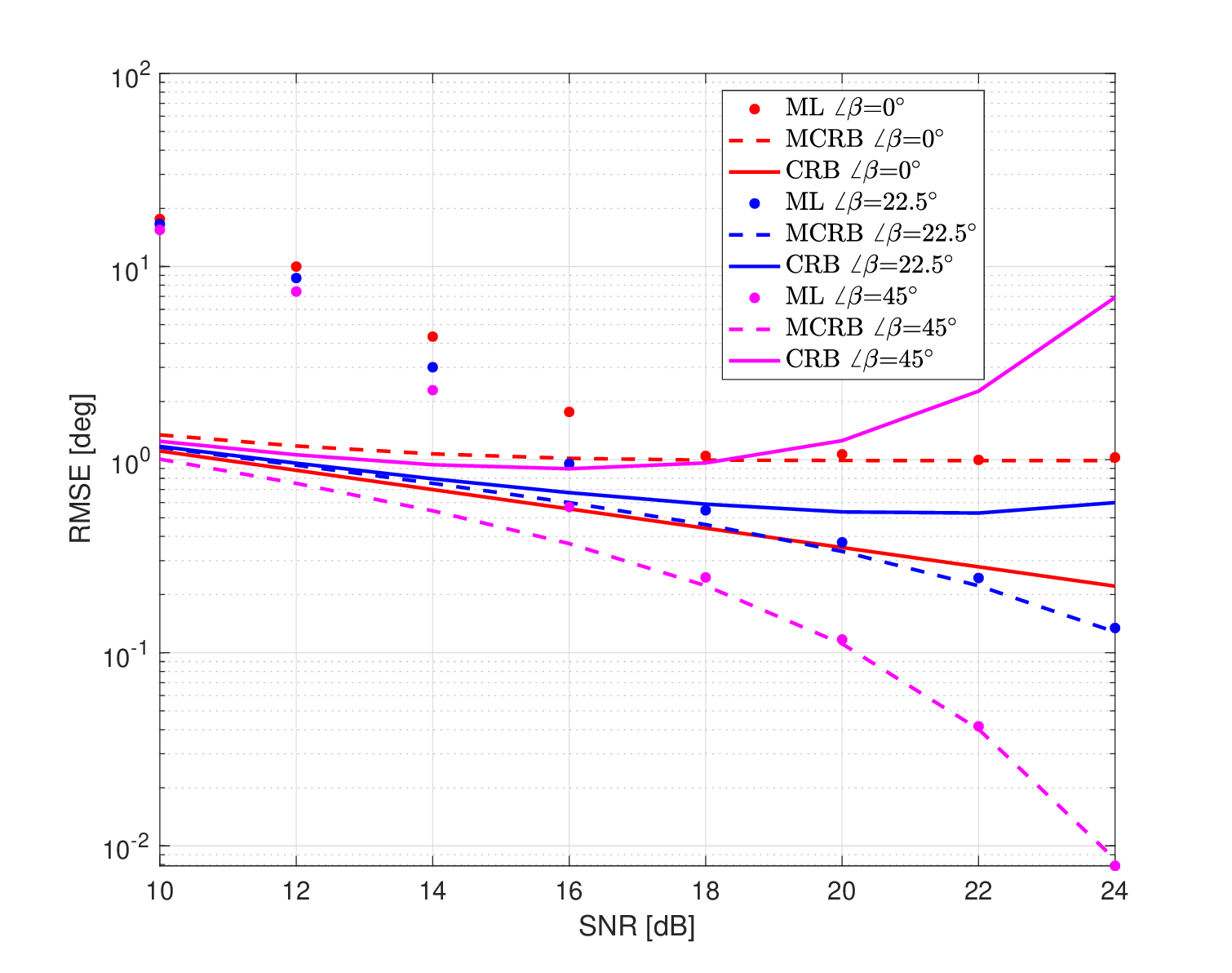}}
\caption{RMSE of misspecified \ac{ML}, \ac{MSE} bound, and \ac{CRB} versus \ac{SNR} for several complex amplitude phases, \ac{ULA} of $N=16$ sensors, and $\varphi=0^\circ$.}
\label{fig 4 MCRB}
\end{figure}
Fig. \ref{fig 4 MCRB} shows that the misspecified \ac{ML} estimator is bounded by (\ref{eq: MSE bound}), and its \ac{RMSE} approaches it for high \ac{SNR}'s.
The \ac{MCRB} for complex amplitude phase $\angle{\beta}=0^\circ$ converges to a constant value
at high \ac{SNR}s.
This phenomenon and the dependency on complex amplitude phase $\angle{\beta}$ is explained as follows.
Substituting $\varphi=0^\circ$, $\angle{\beta}=0^\circ$, and (\ref{eq: DOA signal vector})  in (\ref{eq: x model}), and evaluating the real and imaginary parts of its $n$-th entry result in
\begin{align} \label{eq: x_n real signal and theta=0}
\begin{split}
    x_{n,R} &= \beta \left[\avec(0)\right]_n+v_{n,R},\\ 
    x_{n,I} &= v_{n,I},
\end{split}
\end{align}
since $\beta \left[\avec(0)\right]_n$ is real.
Substituting (\ref{eq: x_n real signal and theta=0}) in (\ref{eq: z_n}) gives the quantized measurements
\begin{align}\label{eq: z_n real signal and theta=0}
    z_n = \text{sign}\left(\beta \left[\avec(0)\right]_n + v_{n,R} \right)+j\cdot \text{sign}\left(v_{n,I}\right),
\end{align}
as the complex amplitude and DOA terms are present in the real part and absent in the imaginary part. Thus, the imaginary part of (\ref{eq: z_n real signal and theta=0}) includes only noise, and according to (\ref{eq: z_n probabilites 2}), it is symmetric Bernoulli-distributed.
As \ac{SNR} increases in (\ref{eq: z_n real signal and theta=0}), the real part converges to $1$.
Evaluating the possible phases of the measurement $z_n$ gives $\angle{z_n} \in \left\{-45^{\circ},45^{\circ} \right\}$, where each one of the two values is obtained \ac{w.p.} $0.5$.
Since the \ac{DOA} information is induced in the measurement phases, by the uncertainly in the phase $\angle{z_n}$ one obtains constant \ac{RMSE} performance as shown in Fig. \ref{fig 4 MCRB}. 

For other complex amplitude phases like $\angle{\beta}=45^{\circ}$, following the same derivation in (\ref{eq: x_n real signal and theta=0})-(\ref{eq: z_n real signal and theta=0}) concludes that the real and imaginary parts of $z_n$ include both signal and noise parts. Thus, as \ac{SNR} increases, the measurement phase $\angle{z_n}$ converges to a specific value \ac{w.p.} $1$, and the \ac{RMSE} reduces, as shown in Fig. \ref{fig 4 MCRB}. 
Comparing the \ac{MCRB} for the misspecified scenario where the quantization is ignored to the quantized \ac{CRB} gives an interesting 
insight. For some complex amplitude phases like $\angle{\beta}=45^\circ$, using conventional algorithms which ignore the quantization, decreases the estimation errors compared to quantized estimation, while for other phases like $\angle{\beta}=0^\circ$ the estimation errors are increased. The \ac{MCRB} and \ac{CRB} differences can be explained by considering the classes of estimators that the bounds apply to. The \ac{MCRB} applies to MS-unbiased estimators, while the \ac{CRB} applies to unbiased estimators. The class of MS-unbiased estimators are generally mean-biased in the perfectly specified sense. Thus, their \ac{RMSE} performance is not bounded by the \ac{CRB}, and may be lower than the \ac{CRB}. To emphasize this phenomenon,
Fig. \ref{fig:mainfigure 2} depicts the \ac{CRB} and \ac{MCRB} including bias versus \ac{DOA} and complex amplitude phase $\angle{\beta}$ for $SNR=20$dB. Notice that both bounds depend on complex amplitude phase, are symmetric around $\angle{\beta}=45^\circ$ and periodic with period $\angle{\beta}=90^\circ$ due to the symmetry of the real and imaginary parts of $\zvec$.  
For example, for \ac{DOA}s $\varphi=0^\circ, -30^\circ, 30^\circ$, the \ac{MCRB} and \ac{CRB} increase or decrease as function of $\angle{\beta}$. 
For these \ac{DOA}s, the \ac{CRB} reaches significantly higher values than the \ac{MCRB} for most complex amplitude phases. Thus, surprisingly, ignoring the quantization in estimation can decrease the estimation errors.
\begin{figure}[htbp]
    \centering
    \begin{subfigure}{0.5\textwidth}
        \centering
\includegraphics[width=\textwidth]{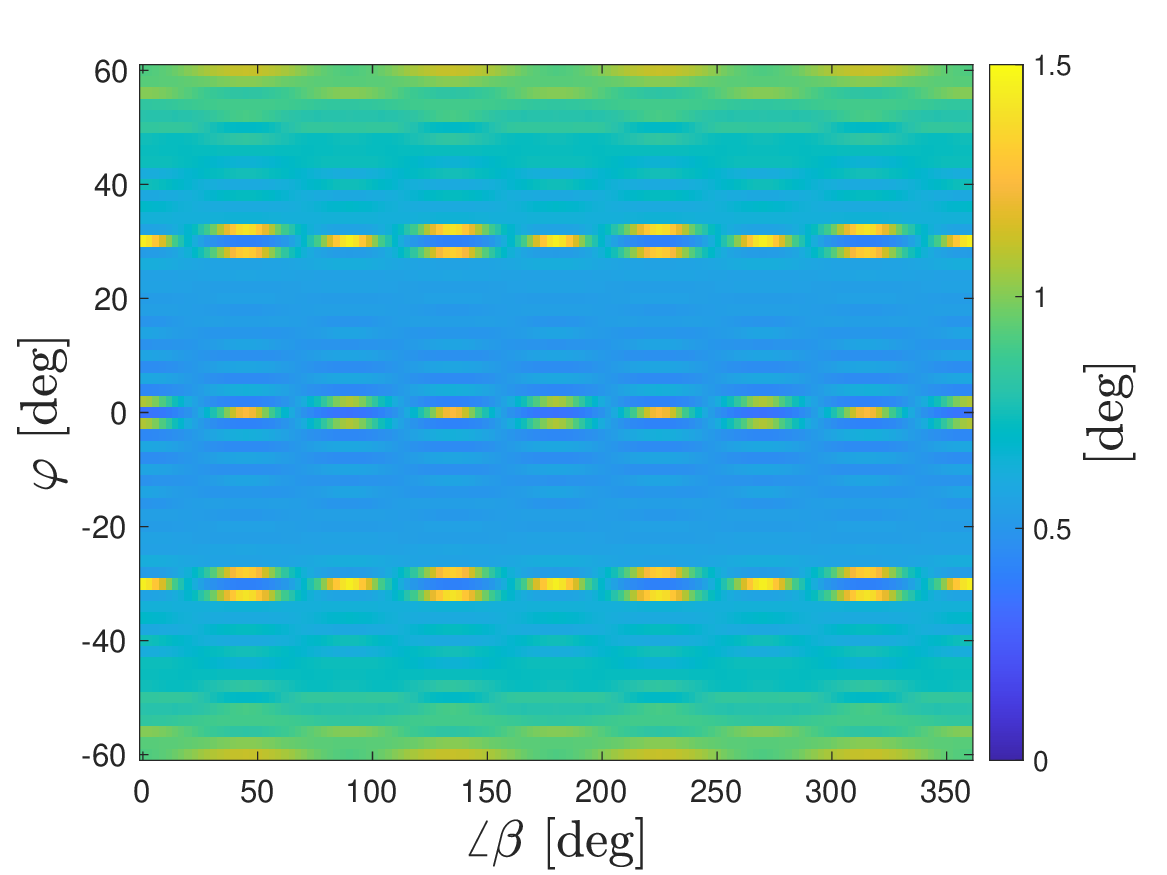}
\caption{\ac{CRB}}
\label{fig 5 CRB}
\end{subfigure}%
\vspace{0.5cm}
\begin{subfigure}{0.5\textwidth}
    \centering
\includegraphics[width=\textwidth]{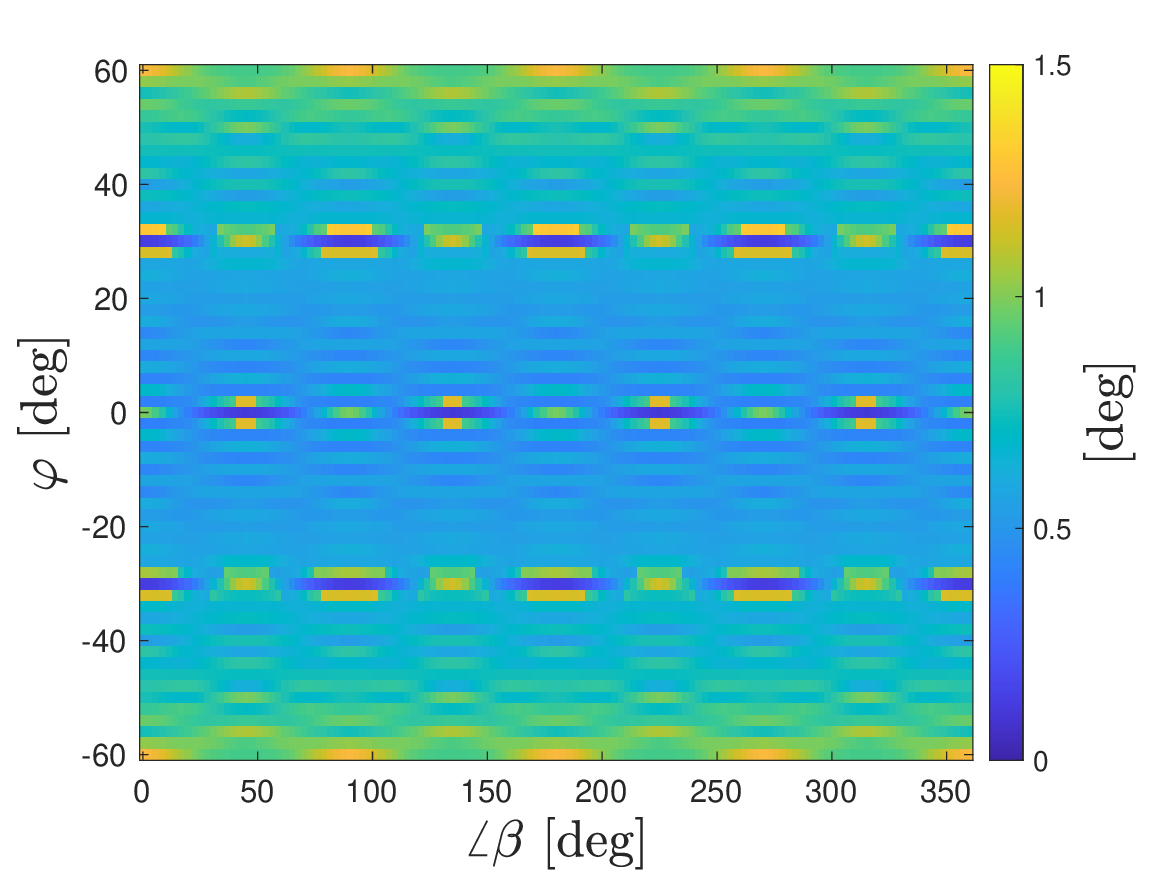}
\caption{\ac{MCRB} including bias}
\label{fig 5 MCRB}
\end{subfigure}
\caption{Comparison of the \ac{CRB} and the \ac{MCRB} including bias contribution versus \ac{DOA} and complex amplitude phase $\angle{\beta}$ for \ac{ULA} of $N=16$ sensors, and $SNR=20$dB.}
    \label{fig:mainfigure 2}
\end{figure}

\subsection{Oversampling in Frequency Estimation}
\label{sec: frequency estimation}
In the problem of frequency estimation, the measurements model before quantization can be described by (\ref{eq: x model}) and (\ref{eq: DOA signal vector}), where the elements of the vector $\avec(\varphi)$ are
\begin{align}
    \label{eq: a frequency}
    \left[\avec(\varphi)\right]_n = \frac{1}{\sqrt{N}} \exp\left[j 2 \pi \varphi t_n\right], \quad n=1,\dots,N,
\end{align}
$\varphi \in \Omega_{\varphi} \subseteq [0,2.5\text{KHz})$ is the frequency, and $\{t_n\}$ are the time samples
\begin{align}
    t_n = t_1 + (n-1) T_s, \quad n=1,\dots,N.
\end{align}
The time samples are symmetric around $0$ \ac{s.t.} $t_1=-\frac{T}{2}$, $t_N=\frac{T}{2}$ where $T$ is the observation interval, and the assumption $\Dot{\avec}^H(\varphi) \avec(\varphi)=0, \quad \forall \varphi \in \Omega_\varphi$ is satisfied. 
As the signal frequency is unknown, we specify that the ideal low-pass filter bandwidth is $B=2.5$KHz. The following figures illustrate the estimation performance of one-bit quantized data after oversampling. 
We denote the oversampling factor by $U\triangleq \frac{f_s}{B}$. 

Figs. \ref{fig 1 AF freq} and \ref{fig 2 MAF freq} depict the \ac{AF} and \ac{MAF}, respectively, versus frequency $\varphi$ for sampling rate $f_s=2.5$KHz, observation interval $T=8$msec, real $\beta$, and \ac{SNR} per sample of $7$dB.
The \ac{MAF} for frequency estimation that considers the misspecification due to one-bit quantization presents similar properties to the \ac{MAF} for \ac{DOA} estimation in Fig. \ref{fig 2 MAF}.
Figs. \ref{fig 3 MAF freq} and \ref{fig 4 MAF freq} depict the \ac{MAF}s for oversampled signal with oversampling factor $U$. Compared to the \ac{MAF} in Fig. \ref{fig 2 MAF freq}, as the oversampling factor $U$ increases, the distortions in the \ac{MAF} that result from the one-bit quantization decay and vanish. For example, the three sidelobe patterns in Fig. \ref{fig 2 MAF freq} reduce to one sidelobe in Fig. \ref{fig 3 MAF freq}, and none in \ref{fig 4 MAF freq}. Thus, oversampling reduces the effects of quantization on the \ac{MAF}. Recalling the correlator $c\left(\avec(\varphi'),\zvec\right)$ and evaluating its expectation gives
\begin{align} \label{eq: correlator z expectation}
    \begin{split}
\E\left[c\left(\avec(\varphi'),\zvec\right)\right] &= \avec^H(\varphi') \muvec(\thetavec)= C \cdot MAF\left(\varphi',\thetavec\right), 
    \end{split}
\end{align}
where $C$ is a constant in $\varphi'$. 
As mentioned above, the \ac{MAF} includes distortions and sidelobe patterns, which their influence is reduced with increased oversampling factor $U$. Thus, oversampling reduces the misspecified \ac{ML} $\hat{\varphi}(\zvec)$ sensitivity to sidelobes and distortions caused by quantization, and improves the estimation accuracy.

\begin{figure}[htbp]
    \centering
\begin{subfigure}[b]{0.24\textwidth}
        \centering
\includegraphics[width=\textwidth]{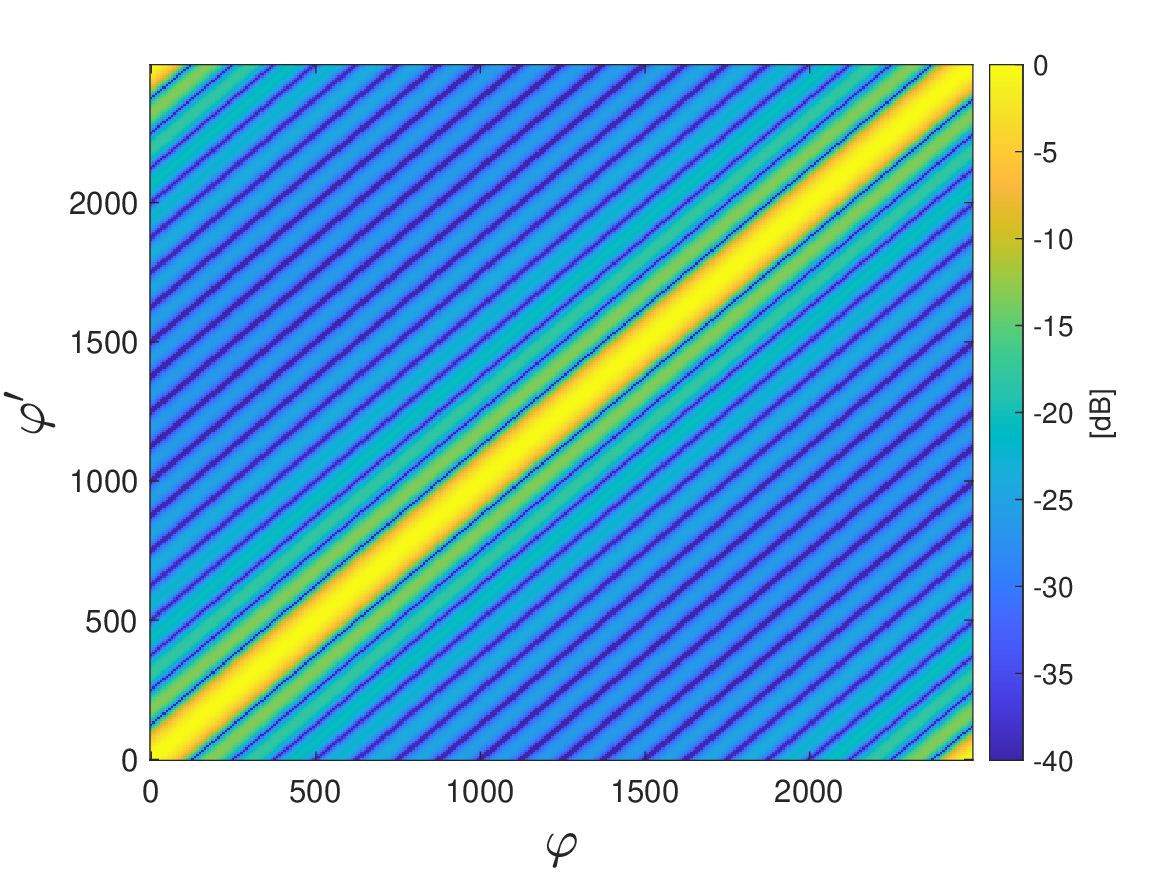}
\caption{$\left|AF\right|$}
\label{fig 1 AF freq}
\end{subfigure}%
\begin{subfigure}[b]{0.24\textwidth}
    \centering
\includegraphics[width=\textwidth]{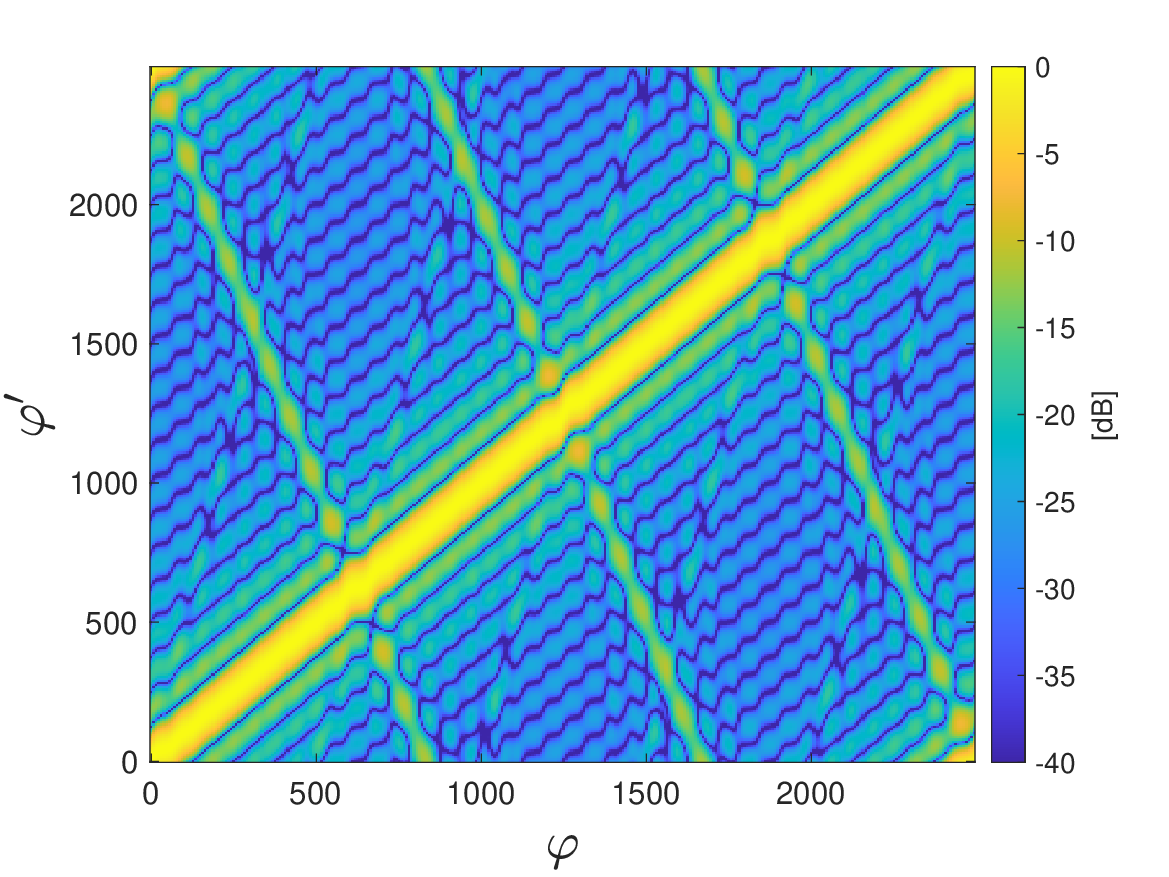} \caption{$\left|MAF\right|$}
\label{fig 2 MAF freq}
\end{subfigure}
\caption{Comparison of ambiguity functions for frequency estimation for sampling rate $f_s=2.5$KHz, observation interval $T=8$msec,  complex amplitude phase $\angle{\beta}=0^\circ$, and $SNR=7$dB per sample.} \label{fig:AF and MAF for frequency}
\end{figure}

\begin{figure}[htbp]
    \centering
\begin{subfigure}[b]{0.24\textwidth}
        \centering
\includegraphics[width=\textwidth]{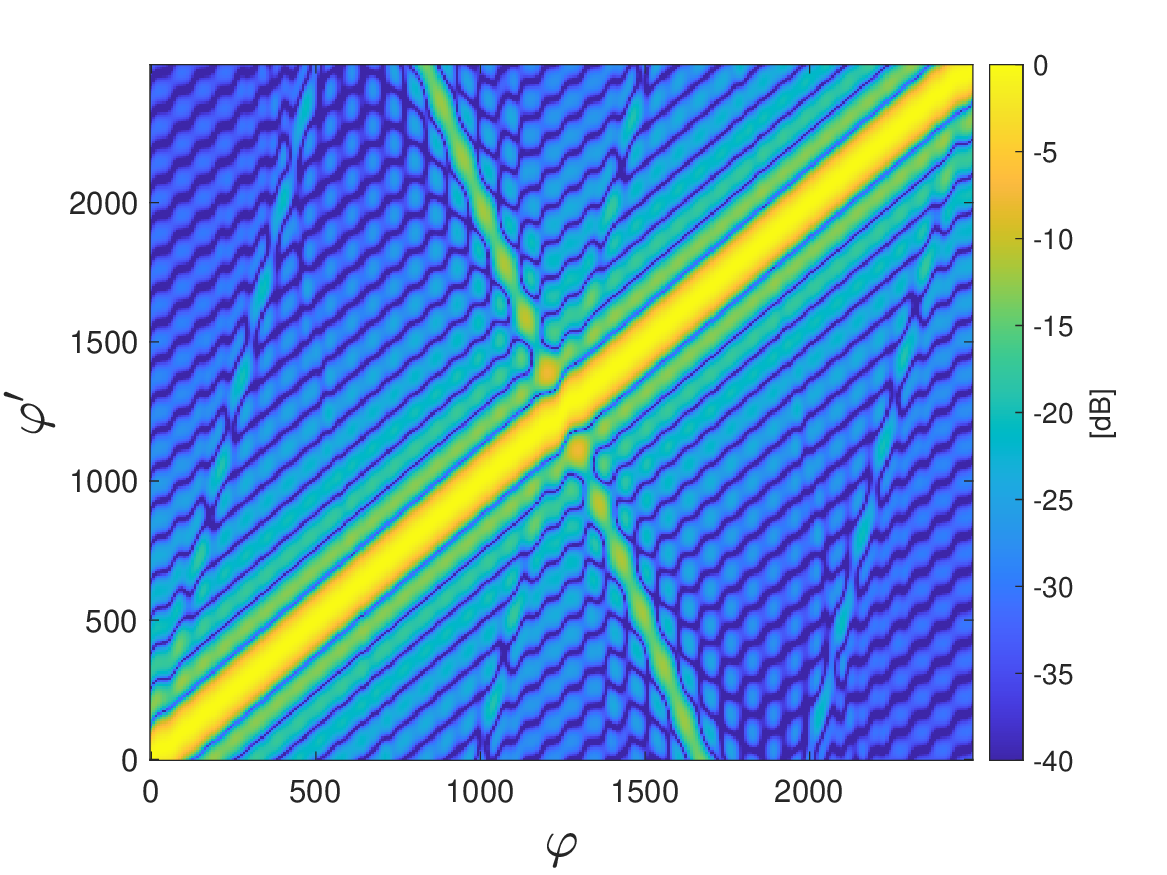}
\caption{$\left|MAF\right|, U=2$ }
\label{fig 3 MAF freq}
\end{subfigure}%
\begin{subfigure}[b]{0.24\textwidth}
    \centering
\includegraphics[width=\textwidth]{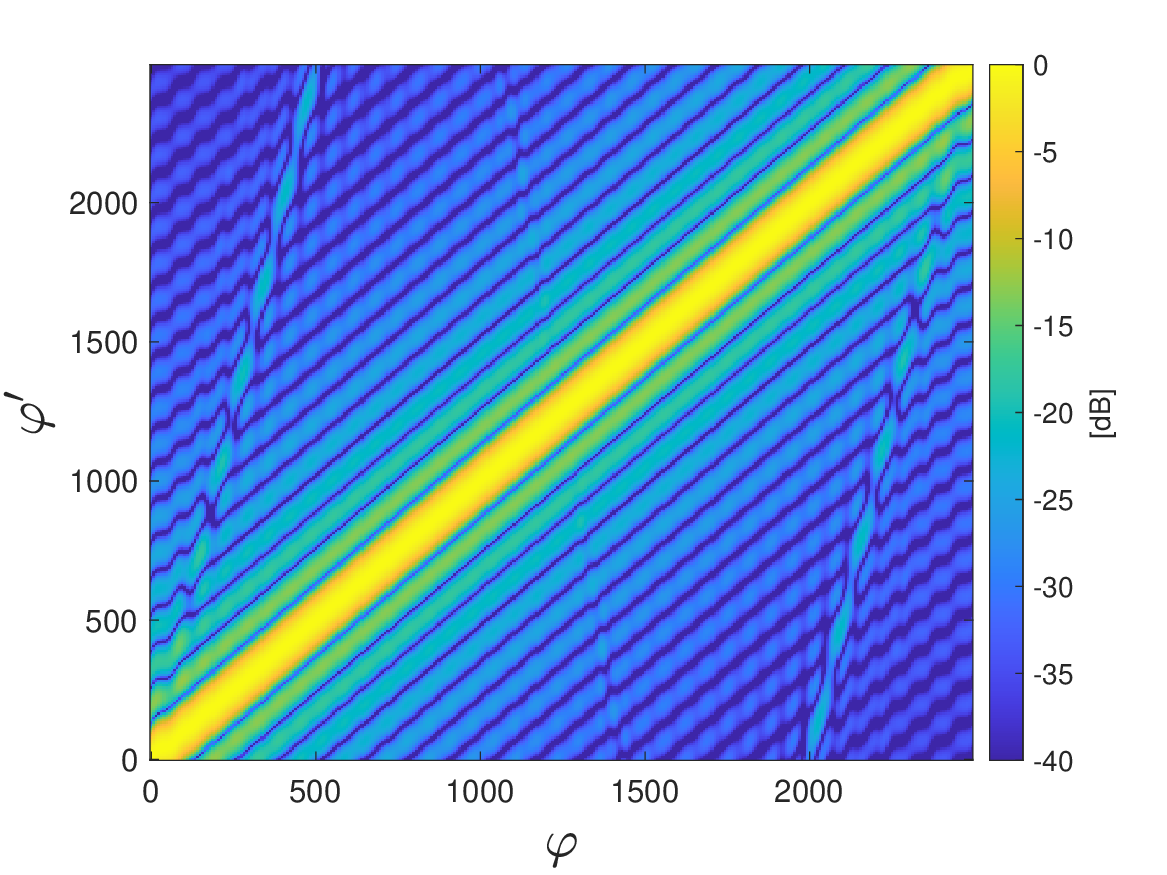} \caption{$\left|MAF\right|, U=4$}
\label{fig 4 MAF freq}
\end{subfigure}
\caption{Comparison of \ac{MAF}s for oversampled frequency estimation for sampling rate $f_s=2.5$KHz $\cdot U$, observation interval $T=8$msec, complex amplitude phase $\angle{\beta}=0^\circ$, and $SNR=7$dB per sample.} \label{fig:MAFs oversampled frequency}
\end{figure}
The \ac{RMSE}s of the \ac{ML} estimator with fine-quantized measurements (\ref{eq: MML x}) are presented in Fig. \ref{fig 6: ML before quantization} versus \ac{SNR} per sample for several oversampling factors $U$, with frequency $\varphi=1$KHz, observation interval $T=4$msec, and complex amplitude phase $\angle{\beta}=60^\circ$. It can be seen that the \ac{RMSE}s of the three estimators coincide in the asymptotic region. Thus, oversampling of fine-quantized band-limited measurements does not reduce the \ac{RMSE} for large \ac{SNR}s. This result is expected, as oversampling does not improve band-limited signal reconstruction for sampling frequencies above the Nyquist rate.
\begin{figure}[htbp]
\centerline{\includegraphics[width=0.5\textwidth]{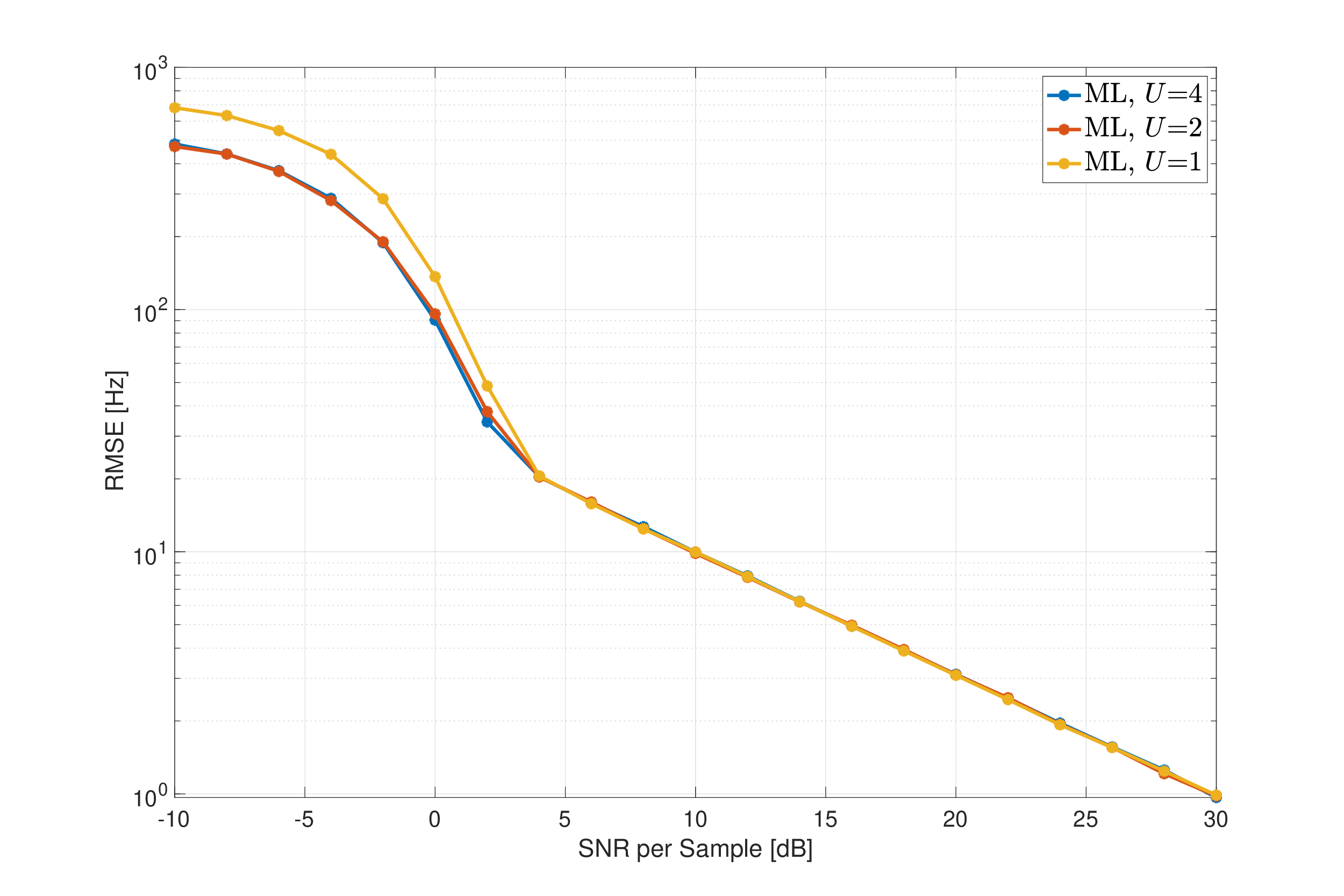}}
\caption{\ac{RMSE} of the \ac{ML} estimator with fine quantization measurements versus \ac{SNR} for sampling rate $f_s=2.5$KHz$ \cdot U$ and oversampling factors $U$, frequency $\varphi=1$KHz, observation interval $T=4$msec, and complex amplitude phase $\angle{\beta}=60^\circ$.}
\label{fig 6: ML before quantization}
\end{figure}
The \ac{RMSE}s of the \ac{ML} estimator with fine quantization measurements, $\hat{\varphi}(\xvec)$, and with misspecified quantized measurements, $\hat{\varphi}(\zvec)$, are presented in Fig. \ref{fig 7: ML before quantization and MML after quantization} versus \ac{SNR} per sample for several oversampling factors $U$, with frequency $\varphi=1$KHz, observation interval $T=4$msec, and complex amplitude phase $\angle{\beta}=60^\circ$. 
\begin{figure}[htbp]
\centerline{\includegraphics[width=0.5\textwidth]{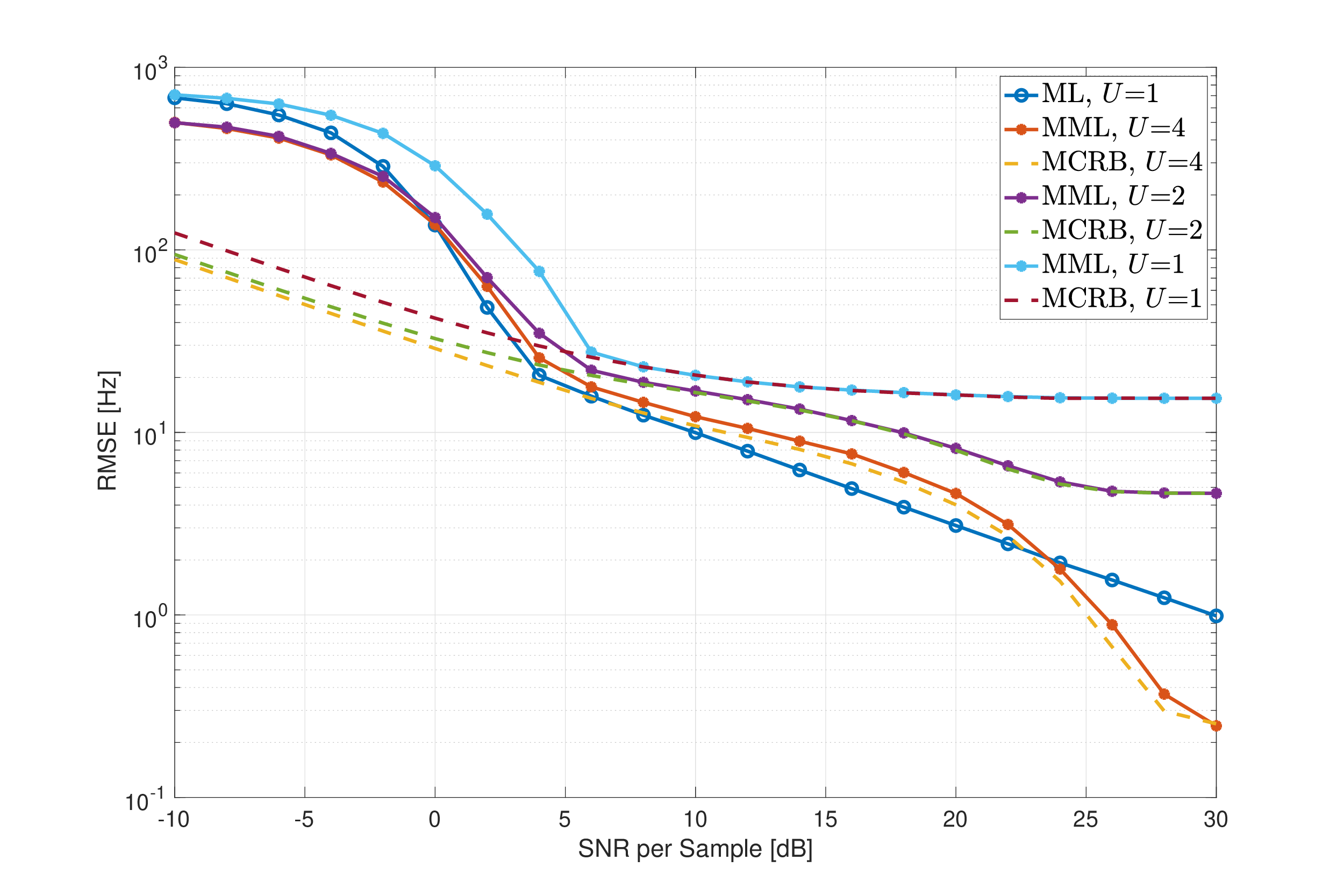}}
\caption{\ac{RMSE}s of the \ac{ML} estimator with fine quantization, misspecified \ac{ML} estimator of the quantized measurements, and \ac{MSE} bound versus \ac{SNR} for sampling rate $f_s=2.5$KHz $\cdot U$ and oversampling factors $U$, with frequency $\varphi=1$KHz, observation interval $T=4$msec, and complex amplitude phase $\angle{\beta}=60^\circ$.}
\label{fig 7: ML before quantization and MML after quantization}
\end{figure}
The \ac{RMSE} of the \ac{ML} estimator in (\ref{eq: MML x}) is presented only for $U=1$, because Fig. \ref{fig 6: ML before quantization} illustrated that oversampling of fine quantized measurements does not reduce the \ac{RMSE}. Moreover, the \ac{MSE} bound for the misspecified estimators due to the quantization is depicted for several oversampling factors in Fig. \ref{fig 7: ML before quantization and MML after quantization}. It can be observed that the \ac{RMSE} of the misspecified \ac{ML} $\hat{\varphi}(\zvec)$ reduces significantly as the oversampling factor $U$ increases for high \ac{SNR}s. The \ac{RMSE} performance of the misspecified \ac{ML} without oversampling ($U=1$) results in large bias for high \ac{SNR}s. 
When the oversampling factor increases to $U=4$, the improvement in \ac{RMSE} performance of the misspecified \ac{ML} is most noticeable, with smallest expected bias and \ac{RMSE} lower than that of the \ac{ML} that relies on fine quantization measurements. The influence of oversampling on estimation performance is verified by the \ac{MSE} bound, which predicts the
asymptotic \ac{RMSE} of the misspecified \ac{ML}. 
Thus, oversampling can significantly improve the estimation performance when the measurements are quantized, even outperforming the estimation performance with fine-quantized measurements. One might ask how using oversampled and quantized measurements can improve the \ac{RMSE} performance compared to using fine quantization measurements. 
The misspecified \ac{ML} estimator based on quantized measurements belongs to the class of MS-unbiased estimators, which are generally biased in the perfectly specified case. Thus, the \ac{RMSE} of the misspecified \ac{ML} estimator is not bounded by the conventional \ac{CRB} based on fine quantization measurements, and can be lower than the \ac{CRB}. 
Recalling the discussion in Section \ref{sec: algorithms}, oversampling with factors $U=2,4$ and estimation with the one-bit quantization is more efficient or has similar computational costs compared to estimation based on fine quantization measurements. Thus, the estimation procedure $\hat{\varphi}(\zvec)$ after oversampling can gain better estimation performance compared to the estimation procedure (\ref{eq: MML x}) without significant additional computational costs. 

In order to investigate the influence of oversampling on estimation performance, the \ac{MSE} bound (\ref{eq: MSE bound}) is presented in Fig. \ref{fig 8: MCRB+bias MML oversampled} versus complex amplitude phase $\angle{\beta}$ for several oversampling factors $U$, with frequency $\varphi=1$KHz, and observation interval $T=4$msec. For most complex amplitude phases $\angle{\beta}$, increasing the oversampling factor $U$ reduces the \ac{MSE} performance, where significant reductions appear for oversampling factors $U=2,4$. However, for complex amplitude phases $\angle{\beta}=45^\circ,135^\circ,225^\circ,315^\circ$, the \ac{MSE} performance for $U=1$ is relatively low. Therefore, \ac{MSE} performance strongly depends on the complex amplitude phase $\angle{\beta}$ as well.
\begin{figure}[htbp]
\centerline{\includegraphics[width=0.5\textwidth]{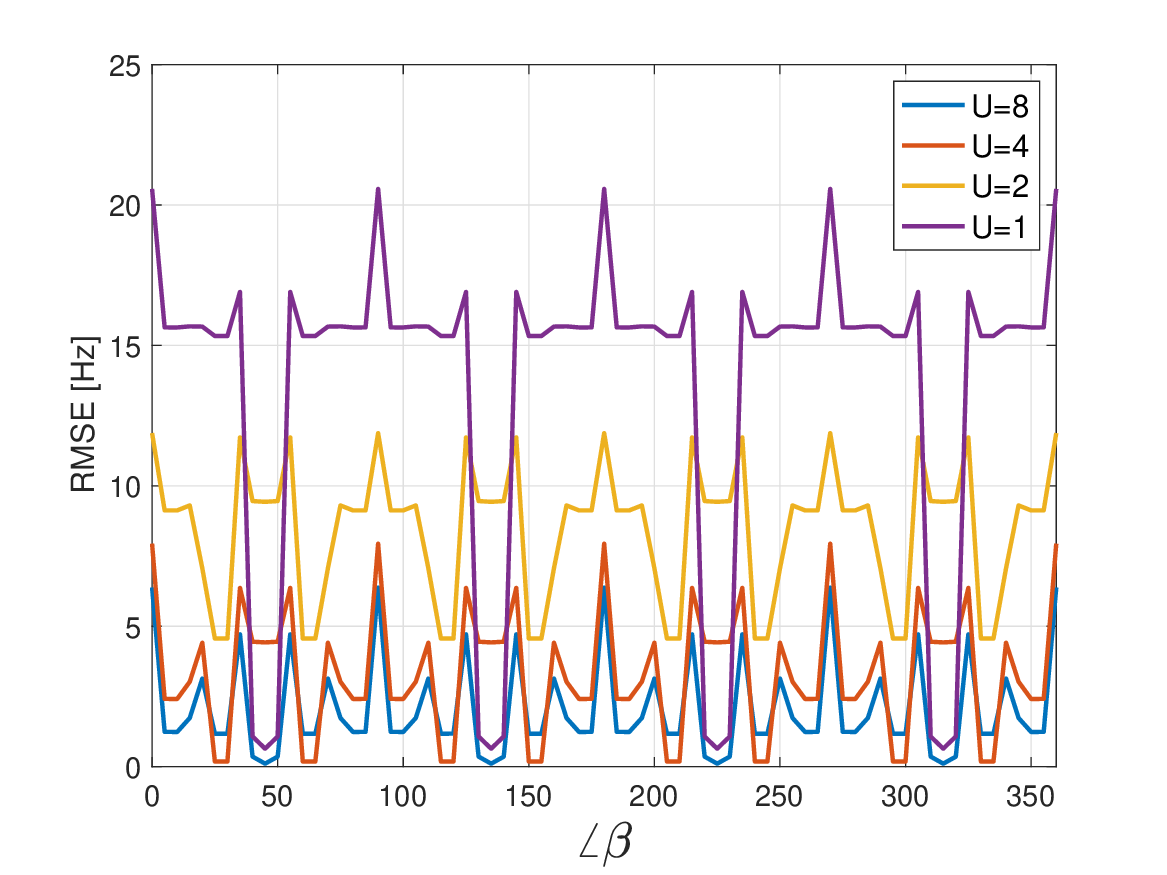}}
\caption{\ac{MSE} bound versus complex amplitude phase $\angle{\beta}$, \ac{SNR}=$30$dB per sample, for sampling rate $f_s=2.5$KHz $\cdot U$ and oversampling factors $U$, frequency $\varphi=1$KHz, and observation interval $T=4$msec.}
\label{fig 8: MCRB+bias MML oversampled}
\end{figure}

\section{Conclusion}
\label{sec: Conclusion}
In this paper, we considered the problem of parameter estimation with misspecified one-bit quantized measurements. Instead of regarding the quantization in the estimation procedure, which results in a complex estimator, one may ignore the quantization and use conventional estimation procedures, introducing misspecification in the model. In order to investigate the estimation performance, the \ac{MCRB} and expected bias are derived for quantized measurements with \ac{AWGN} or colored noise. A comparison of the computational cost of estimation procedures with fine quantization and one-bit measurements is given. 
Test cases of \ac{DOA} and frequency estimation are presented to investigate the influence of misspecification due to quantization. 
For the problem of \ac{DOA} estimation, the 
\ac{RMSE} performance of the misspecified \ac{ML} estimator and the \ac{MCRB} depend on \ac{DOA}, \ac{SNR}, and complex amplitude phase.
For the problem of frequency estimation, the influence of oversampling on quantized measurements is investigated. It is shown that oversampling reduces the effects of quantization on estimation performance, and can result in better \ac{MSE} performance compared to the \ac{MSE} performance of estimation based on fine quantization measurements.

The \ac{MCRB} derived in this paper provides a useful tool for analysis of misspecified one-bit quantization. Future work can investigate the expected performance as a function of different quantization parameters, such as quantization offset. Moreover, the \ac{MCRB} can be used to design optimal time-varying offset or optimal sampling time in non-uniform sampling. Another research directions can focus on exploring the effect of threshold region shift due to misspecified quantization. It can be done via derivation of misspecified Barankin-type bounds and their implementation for misspecified quantized model with oversampling. 
\newline
\section*{Appendix A \\ Statistics of Quantized Measurements} 
\addcontentsline{toc}{section}{Appendix A} %
\label{Appendix: statistics}
The expectation of the quantized measurement $z_n$ can be derived by using (\ref{eq: z_n probabilites 1})-(\ref{eq: q_n}) to obtain
\begin{align} \label{eq: mu n}
    \mu_n (\thetavec) \triangleq \E \left[z_n\right],
\end{align}
where $\muvec(\thetavec) \triangleq \left[\mu_1(\thetavec),\dots,\mu_N(\thetavec)\right]^T$,
\begin{align} \label{eq: mu R n}
\begin{split}
    \mu_{n,R}(\thetavec) &\triangleq \E \left[z_{n,R}\right] = \sum_{z_{n,R} \in \{-1,1\}} z_{n,R} p_{z_{n,R};\thetavec}(z_{n,R};\thetavec)\\
    &=- Q\left(- q_{n,R}\left(\thetavec\right) \right) +Q\left(q_{n,R}\left(\thetavec\right) \right)\\
    &= 1-2Q\left(- q_{n,R}\left(\thetavec\right) \right),\\
\end{split}
\end{align}
in which the property $Q(-x) = 1 -Q(x)$ is used, and
similarly
\begin{align}
    \label{eq: mu I n}
\begin{split}
\mu_{n,I}(\thetavec) &\triangleq \E \left[z_{n,I}\right] = \sum_{z_{n,I} \in \{-1,1\}} z_{n,I} p_{z_{n,I};\thetavec}(z_{n,I};\thetavec)\\
&=  1-2Q\left(- q_{n,I}\left(\thetavec\right)\right).
    \end{split}
\end{align}
 As $z_{n,R}, z_{n,I} \in \{-1,1\}$, it can be concluded that $\E \left[z_{n,R}^2\right] = 1, $ and $\E \left[z_{n,I}^2\right] = 1$, 
and therefore 
\begin{align}\label{eq: z abs and squared}
\begin{split}
    \E \left[\left|z_n\right|^2\right] &= \E \left[z_{n,R}^2 + z_{n,I}^2\right] = 2,\\
     \E \left[z_n^2\right] &= \E \left[z_{n,R}^2 - z_{n,I}^2 + 2j z_{n,R} z_{n,I} \right] \\
     &= 2j \cdot \mu_{n,R}(\thetavec) \mu_{n,I}(\thetavec),
\end{split}
    \end{align}
where the last equality in (\ref{eq: z abs and squared}) is obtained because $z_{n,R}$ and $z_{n,I}$ are uncorrelated.

\section*{Appendix B \\ Proof of Proposition 1} 
\addcontentsline{toc}{section}{Appendix B} %
\label{Appendix: proof} 
The search parameter $\varphi'$ is a value in a finite grid with $K$ possible values. 
For a grid parameter $\varphi'$, the correlator inside the argument in (\ref{eq: MML x}) is
\begin{align} \label{eq: correlator x}
\begin{split}
c\left(\avec(\varphi'),\xvec\right) =&\sum_{n=1}^N  a^*_n(\varphi') x_n,
\end{split}
\end{align}
where the $n$-th element in the summation in (\ref{eq: correlator x}) results in
\begin{align}
    \label{eq: element n correlator x}
    \begin{split}
a^*_n(\varphi') x_n=& \left(a_{n,R}(\varphi') -ja_{n,I}(\varphi')\right) \left(x_{n,R}+j x_{n,I}\right)\\
=& \left(x_{n,R}a_{n,R}(\varphi') +x_{n,I} a_{n,I}(\varphi')\right)\\\
&+ j \left(x_{n,I} a_{n,R}(\varphi') -x_{n,R}a_{n,I}(\varphi') \right).  
    \end{split}
\end{align}
As the real and imaginary parts $x_{n,R}$ and $x_{n,I}$ satisfy $x_{n,R},x_{n,I} \in \mathbb{R}$, the computational cost of (\ref{eq: element n correlator x}) results in 4 multiplications and 2 additions, with cost $6C_{op}$.   

The correlator $c\left(\avec(\varphi'),\zvec\right)$ can be defined similarly to (\ref{eq: correlator x}) with the quantized measurements $\zvec$. The $n$-th element in the summation of the correlator $c\left(\avec(\varphi'),\zvec\right)$ 
gives four possible results:
\begin{align}\label{eq: element n correlator results}
\begin{split}
    a^*_n(\varphi') z_n \in \Big \{& f_n(\varphi')+j \cdot d_n(\varphi'),\: d_n(\varphi')-j \cdot f_n(\varphi'),\\
    &-d_n(\varphi')+j \cdot f_n(\varphi'), \: -f_n(\varphi')-j \cdot d_n(\varphi') \Big\},  
\end{split}
\end{align}
where 
\begin{align} \label{eq: a_n}
\fvec(\varphi')\triangleq \Re\left(\avec(\varphi')\right)+\Im\left(\avec(\varphi')\right)
\end{align} and
\begin{align}\label{eq: d_n}
\dvec(\varphi') \triangleq \Re\left(\avec(\varphi')\right)-\Im\left(\avec(\varphi')\right),
\end{align}
because $z_{n,R},z_{n,I} \in \{-1,1\}$. 
The four possible terms in (\ref{eq: element n correlator results}) are evaluated using additions or subtractions (without multiplications) in the elements of (\ref{eq: a_n}) and (\ref{eq: d_n}). Specifically, the computational cost of (\ref{eq: element n correlator results}) is four additions/subtractions to yield the values $\{\pm f_n(\varphi'),\pm d_n(\varphi')\}$. Thus, the possible values in (\ref{eq: element n correlator results}) can be evaluated in cost of $4C_{op}$. 

The computational cost to evaluate each possible result of $a^*_n(\varphi') z_n$ for every quantized sample $z_n, \quad n=1,\dots,N,$ and each grid parameter $\varphi'$ is therefore $4C_{op}NK$. Those computations can be executed once in an offline procedure and saved in a look-up table.
Thus, those computations are negligible in terms of computational cost.

To evaluate the complex summations in the correlators (\ref{eq: correlator x}) and $c\left(\avec(\varphi'),\zvec\right)$, $2(N-1)$ additions are required. Two additional multiplications and one sum are required to evaluate the squared absolute values of the correlators (\ref{eq: correlator x}) and $c\left(\avec(\varphi'),\zvec\right)$. 

To summarize, the computational costs for evaluating the squared absolute values of the correlators $c\left(\avec(\varphi'),\zvec\right)$ and (\ref{eq: correlator x}) for each grid parameter $\varphi'$ are
\begin{align} \label{eq: Cz}
\begin{split}
    C_{\zvec} =&\left(2(N-1) + 3 \right) C_{op} K\\
    =& \left(2N+1\right)C_{op}K
\end{split}
    \end{align}
and 
\begin{align} \label{eq: Cx}
\begin{split}
    C_{\xvec} &= C_{op} \left(6N + 2(N-1) + 3 \right) K\\
    &=C_{op} \left(8 N + 1 \right) K,
\end{split}
\end{align}
respectively.
The reduction in computational complexity obtained by using one-bit measurements instead of fine quantization measurements is given by the ratio of (\ref{eq: Cz}) and (\ref{eq: Cx}),
which gives (\ref{eq: CU approx}).

\bibliographystyle{IEEEtran}
\bibliography{strings}

\end{document}

%% file: shortcuts.tex
\newcommand{\avec}{{\bf{a}}}
\newcommand{\bvec}{{\bf{b}}}

\newcommand{\dvec}{{\bf{d}}}

\newcommand{\fvec}{{\bf{f}}}

\newcommand{\uvec}{{\bf{u}}}

\newcommand{\xvec}{{\bf{x}}}
\newcommand{\zvec}{{\bf{z}}}

\newcommand{\svec}{{\bf{s}}}
\newcommand{\vvec}{{\bf{v}}}

\newcommand{\hvec}{{\bf{h}}}

\newcommand{\onevec}{{\bf{1}}}
\newcommand{\zerovec}{{\bf{0}}}

\newcommand{\Amat}{{\bf{A}}}
\newcommand{\Bmat}{{\bf{B}}}
\newcommand{\Cmat}{{\bf{C}}}

\newcommand{\Hmat}{{\bf{H}}}
\newcommand{\Jmat}{{\bf{J}}}

\newcommand{\Imat}{{\bf{I}}}

\newcommand{\Mmat}{{\bf{M}}}
\newcommand{\Pmat}{{\bf{P}}}

\newcommand{\Rmat}{{\bf{R}}}

\newcommand{\Diag}{{\rm Diag}}

\def\thetavec{{\boldsymbol{\theta}}}

\def\muvec{{\mbox{\boldmath $\mu$}}}

\newcommand{\E}{{\rm{E}}}

\newcommand{\be}{\begin{equation}}
\newcommand{\ee}{\end{equation}}
\newcommand{\beqna}{\begin{eqnarray}}
\newcommand{\eeqna}{\end{eqnarray}}

\makeatletter
\def\user@resume{resume}
\def\user@intermezzo{intermezzo}
\newcounter{previousequation}
\newcounter{lastsubequation}
\newcounter{savedparentequation}
\makeatother 
\newcommand{\diff}{{\textnormal{d}}}

\newtheorem{prop}{Proposition}